\def\llncs{0}
\def\fullpage{1}
\def\anonymous{0}
\def\authnote{0}
\def\notxfont{0}
\def\submission{0}

\ifnum\submission=1
\def\llncs{1}
\fi

\ifnum\llncs=1
\documentclass[envcountsect,a4paper,runningheads,10pt]{llncs}
\else
	\documentclass[letterpaper,hmargin=1.05in,vmargin=1.05in,11pt]{article}
			\ifnum\fullpage=1
		\usepackage{fullpage}
		\fi
\fi

\usepackage[%
  colorlinks=true,
  citecolor=blue,
  pagebackref=true
]{hyperref}

\usepackage{amsmath, amsfonts, amssymb, mathtools,amscd,wrapfig,xspace}

\usepackage{amsthm}

\usepackage{lmodern}
\usepackage[T1]{fontenc}
\usepackage[utf8]{inputenc}

\usepackage{etex}

\usepackage{arydshln} 
\usepackage{url}
\usepackage{ifthen}
\usepackage{bm}
\usepackage{multirow}
\usepackage[dvips]{graphicx}
\usepackage[usenames]{color}
\usepackage{xcolor,colortbl} 
\usepackage{threeparttable}
\usepackage{comment}
\usepackage{paralist,verbatim}
\usepackage{cases}
\usepackage{booktabs}
\usepackage{braket}
\usepackage{cancel} 
\usepackage{ascmac} 
\usepackage{framed}
\usepackage{authblk}
\usepackage{pifont}
\usepackage{qcircuit}
\usepackage{tikz}
\usetikzlibrary{cd}
\definecolor{darkblue}{rgb}{0,0,0.6}
\definecolor{darkgreen}{rgb}{0,0.5,0}
\definecolor{maroon}{rgb}{0.5,0.1,0.1}
\definecolor{dpurple}{rgb}{0.2,0,0.65}

\usepackage[capitalise,noabbrev]{cleveref}
\usepackage[absolute]{textpos}
\usepackage[final]{microtype}
\usepackage[absolute]{textpos}
\usepackage{everypage}
\DeclareMathAlphabet{\mathpzc}{OT1}{pzc}{m}{it}

\ifnum\submission=1
     \crefname{appendix}{Suppl}{Supplements} 
\fi
\usepackage{algorithmic}
\usepackage{algorithm}
\usepackage{here}

\usepackage[inline]{enumitem}

\usepackage[normalem]{ulem}

\usepackage{thmtools}
\usepackage{thm-restate}

\newtheoremstyle{thicktheorem}%
{\topsep}
{\topsep}
{\itshape}{}%
{\bfseries}%
{.}
{ }%
{\thmname{#1}\thmnumber{ #2}%
		\thmnote{ (#3)}%
}

\newtheoremstyle{remark}
{\topsep}
{\topsep}
	{}
	{}
	{}
	{.}
	{ }
	{\textit{\thmname{#1}}\thmnumber{ #2}
			\thmnote{ (#3)}%
	}

\ifnum\llncs=0

	\theoremstyle{thicktheorem}
	\newtheorem{theorem}{Theorem}[section]
	\newtheorem{lemma}[theorem]{Lemma}
	\newtheorem{corollary}[theorem]{Corollary}
	\newtheorem{proposition}[theorem]{Proposition}
	\newtheorem{definition}[theorem]{Definition}
	
    \newtheorem{construction}[theorem]{Construction}

	\theoremstyle{remark}
	\newtheorem{claim}[theorem]{Claim}
	\newtheorem{remark}[theorem]{Remark}

    \crefname{claim}{Claim}{Claims}

\else

\fi

	\crefname{theorem}{Theorem}{Theorems}
	\crefname{assumption}{Assumption}{Assumptions}
	\crefname{construction}{Construction}{Constructions}
	\crefname{corollary}{Corollary}{Corollaries}
	\crefname{conjecture}{Conjecture}{Conjectures}
	\crefname{definition}{Definition}{Definitions}
	\crefname{example}{Example}{Examples}
	\crefname{experiment}{Experiment}{Experiments}
	\crefname{counterexample}{Counterexample}{Counterexamples}
	\crefname{lemma}{Lemma}{Lemmata}
	\crefname{observation}{Observation}{Observations}
	\crefname{proposition}{Proposition}{Propositions}
	\crefname{remark}{Remark}{Remarks}
	\crefname{claim}{Claim}{Claims}
	\crefname{fact}{Fact}{Facts}
	\crefname{note}{Note}{Notes}

\ifnum\llncs=1
 \crefname{appendix}{App.}{Appendices}
 \crefname{section}{Sec.}{Sections}
\else
\fi

\ifnum\llncs=1
\renewcommand*{\backref}[1]{}
\else
	\renewcommand*{\backref}[1]{(Cited on page~#1.)}
	\ifnum\notxfont=1
	\else
		\usepackage{newtxtext}
	\fi
\fi

\ifnum\authnote=0  
\newcommand{\mor}[1]{}
\newcommand{\shogo}[1]{}
\newcommand{\takashi}[1]{}
\newcommand{\fuyuki}[1]{}
\newcommand{\amit}[1]{}
\newcommand{\tamer}[1]{}

\else
\newcommand{\mor}[1]{$\ll$\textsf{\color{red} Tomoyuki: { #1}}$\gg$}
\newcommand{\takashi}[1]{$\ll$\textsf{\color{orange} Takashi: { #1}}$\gg$}
\newcommand{\shogo}[1]{$\ll$\textsf{\color{darkgreen} Shogo: { #1}}$\gg$}
\newcommand{\fuyuki}[1]{$\ll$\textsf{\color{darkblue} Fuyuki: { #1}}$\gg$}

\newcommand{\amit}[1]{$\ll$\textsf{\color{violet} Amit: { #1}}$\gg$}
\newcommand{\tamer}[1]{$\ll$\textsf{\color{magenta} Tamer: { #1}}$\gg$}

\DeclareRobustCommand{\Erase}{\bgroup\markoverwith{\textcolor{red}{\rule[.5ex]{2pt}{0.4pt}}}\ULon}

\fi

\ifnum\llncs=1

\spnewtheorem{construction}[theorem]{Construction}{\bfseries}{\itshape}

\fi

\newcommand{\proj}[1]{\ket{#1}\!\bra{#1}}
\newcommand{\ip}[2]{\langle{#1}|{#2}\rangle}
\newcommand{\samp}{\mathpzc{Samp}}
\newcommand{\ver}{\mathpzc{Ver}}
\newcommand{\kgen}{\mathpzc{KeyGen}}
\newcommand{\mint}{\mathpzc{Mint}}
\newcommand{\verify}{\mathpzc{Verify}}

\newcommand{\RR}{\mathbb{R}}
\newcommand{\NN}{\mathbb{N}}
\newcommand{\CC}{\mathbb{C}}
\newcommand{\QQ}{\mathbb{Q}}

\newcommand{\Tr}{\mathrm{Tr}}

\newcommand{\ceil}[1]{\left\lceil{#1}\right\rceil}

\newcommand{\cA}{\mathcal{A}}
\newcommand{\cB}{\mathcal{B}}

\newcommand{\cO}{\mathcal{O}}

\def\makeuppercase#1{
\expandafter\newcommand\csname tl#1\endcsname{\widetilde{#1}}
}

\def\makelowercase#1{
\expandafter\newcommand\csname tl#1\endcsname{\widetilde{#1}}
}

\newcommand{\reduction}{\mathcal{R}}

\newcommand{\secp}{\lambda}
\renewcommand{\sec}{\secp}

\newcommand{\oracle}{\mathcal{O}}

\newcommand*{\algo}[1]{\ensuremath{\mathsf{#1}}}

\newenvironment{boxfig}[2]{\begin{figure}[#1]\fbox{\begin{minipage}{0.97\linewidth}
                        \vspace{0.2em}
                        \makebox[0.025\linewidth]{}
                        \begin{minipage}{0.95\linewidth}
            {{
                        #2 }}
                        \end{minipage}
                        \vspace{0.2em}
                        \end{minipage}}}{\end{figure}}

\newcommand{\bit}{\{0,1\}}

\newcommand{\Ver}{\algo{Ver}}

\newcommand{\TD}{\algo{TD}}

\newcommand{\negl}{{\mathsf{negl}}}

\newcommand{\poly}{{\mathrm{poly}}}

\newcommand{\bin}{\{0,1\}}

\usetikzlibrary{decorations.markings,calc}
\tikzset{
  cross/.style={
    postaction={decorate,decoration={markings,
    mark=at position 0.45 with {\draw[-,line width=1pt] (-10pt,-10pt) -- (10pt,10pt);\draw[-,line width=1pt] (-10pt,10pt) -- (10pt,-10pt);}}}
  }
}

\makeatletter
\DeclareRobustCommand
  \myvdots{\vbox{\baselineskip4\p@ \lineskiplimit\z@
    \hbox{.}\hbox{.}\hbox{.}}}
\makeatother

\newcommand{\Span}{\mathsf{Span}}

\newcommand{\G}{\mathpzc{G}}
\renewcommand{\S}{\mathpzc{S}}
\newcommand{\V}{\mathpzc{V}}

\newcommand{\qefid}{\text{QEFID}\xspace}
\newcommand{\uncl}{\text{UCSGs}\xspace}
\newcommand{\struncl}{\text{strong-UCSGs}\xspace}

\newcommand{\estver}{\badv_{\V}}
\newcommand{\estvm}{\badv_{{\VM}}}
\newcommand{\VM}{\mathpzc{VM}}
\newcommand{\adv}{\mathcal{A}}
\newcommand{\badv}{\mathcal{B}}

\newcommand{\SD}{\mathsf{SD}}

\newcommand{\tomography}{\mathpzc{GentleSearch}}

\newcommand{\Expct}{\mathbb{E}}

\newcommand{\keyspace}{\mathcal{K}}
\renewcommand{\braket}[2]{\langle #1|#2\rangle}
\newcommand{\ketbra}[2]{\ket{#1}\bra{#2}}
\newcommand{\reflect}{R}
\newcommand{\symd}{\mathsf{\lor^{\ell+1}\CC^N}}

\newcommand{\PP}{\textbf{PP}\xspace}

\newcommand{\EE}{\mathbb{E}}
\newcommand{\qpspace}{\textbf{QPSPACE}\xspace}

\renewcommand{\QQ}{\mathbf{M}}

\newcommand{\submitorfull}[2]{

\ifnum\submission=1

    #1
    
\else

    \input{proofs/#2}

\fi

}

\title{
A New World in the Depths of Microcrypt:\\ Separating OWSGs and 
Quantum Money from \qefid 
}

\ifnum\anonymous=1
\ifnum\llncs=1
\author{\empty}\institute{\empty}
\else
\author{}
\fi
\else
\ifnum\llncs=1
\author{
Amit Behera\inst{1}\and Giulio Malavolta\inst{2}\and Tomoyuki Morimae\inst{3}\and Tamer Mour\inst{2}\and Takashi Yamakawa\inst{4,3} 
}
\institute{
Ben-Gurion University, Israel \and
Bocconi University, Italy\and
	Yukawa Institute for Theoretical Physics, Kyoto University, Kyoto, Japan \and NTT Social Informatics Laboratories, Tokyo, Japan 
}
\else
\author[1]{Amit Behera}
\author[2]{Giulio Malavolta}
\author[3]{Tomoyuki Morimae}
\author[4]{Tamer Mour}
\author[5,3]{Takashi Yamakawa}
\affil[1]{{\small Department of Computer Science, Ben-Gurion University of the Negev, Beersheba, Israel}\authorcr{\small behera@post.bgu.ac.il
}}
\affil[2]{{\small Bocconi University, Milan, Italy}\authorcr{\small giulio.malavolta@hotmail.it}}
\affil[3]{{\small Yukawa Institute for Theoretical Physics, Kyoto University, Kyoto, Japan}\authorcr{\small tomoyuki.morimae@yukawa.kyoto-u.ac.jp}}
\affil[4]{{\small Bocconi University, Milan, Italy}\authorcr{\small tamer.mour@unibocconi.it}}
\affil[5]{{\small NTT Social Informatics Laboratories, Tokyo, Japan}\authorcr{\small 
takashi.yamakawa@ntt.com}}

\fi 
\fi

\date{}

\begin{document}

\maketitle

\begin{abstract}
While in classical cryptography one-way functions (OWFs) are widely regarded as the ``minimal assumption'', the situation in quantum cryptography is less clear. Recent works have put forward two concurrent candidates for the minimal assumption in quantum cryptography: One-way state generators (OWSGs), postulating the existence of a hard \emph{search} problem with an efficient verification algorithm, and EFI pairs, postulating the existence of a hard \emph{distinguishing} problem.
Two recent papers [Khurana and Tomer STOC'24; Batra and Jain FOCS'24] showed that OWSGs imply EFI pairs, but the
reverse direction remained open. 

In this work, we give strong evidence that the opposite direction does not hold: We show that there is a quantum unitary oracle relative to which EFI pairs exist but OWSGs do not. In fact, we show a slightly stronger statement that holds also for EFI pairs that output classical bits (QEFID pairs).

As a consequence, we separate, via our oracle, \qefid pairs and one-way puzzles from OWSGs and several other Microcrypt primitives, including efficiently verifiable one-way puzzles and unclonable state generators. In particular, this solves a problem left open in [Chung, Goldin, and Gray Crypto'24].

Using similar techniques, we also establish a fully black-box separation (which is slightly weaker than an oracle separation) between private-key quantum money schemes and QEFID pairs.

One conceptual implication of our work is that the existence of an efficient verification algorithm may lead to qualitatively stronger primitives in quantum cryptography. 

\end{abstract}

\ifnum\submission=1
\else
\clearpage
\newpage
\setcounter{tocdepth}{2}
\tableofcontents
\fi
\newpage

\section{Introduction}

In cryptography, a central question is to determine the minimal assumptions needed to construct various cryptographic primitives. In classical cryptography, the existence of one-way functions (OWFs) is regarded as the \emph{minimal} computational assumption: On the one hand, OWFs imply the existence of a broad range of basic cryptographic primitives~\cite{STOC:LubRac86,FOCS:ImpLub89,STOC:ImpLevLub89}, including pseudorandom generators (PRGs)~\cite{HILL99}, pseudorandom functions (PRFs)~\cite{JACM:GolGolMic86}, commitments~\cite{C:Naor89}, symmetric key encryption (SKE), and digital signatures~\cite{STOC:Rompel90}. On the other hand, OWFs are implied by essentially any non-trivial cryptographic primitive with computational security. This equivalence gives strong evidence that OWFs may be the most basic building block in the classical cryptographic landscape.

In contrast, the situation in quantum cryptography appears to be fundamentally different. Morimae and Yamakawa~\cite{C:MorYam22} and Ananth, Qian, and Yuen \cite{C:AnaQiaYue22} independently initiated the study of quantum cryptography from assumptions potentially weaker than OWFs. Specifically, these works demonstrate constructions of various quantum cryptographic primitives, including commitments, SKE, and digital signatures, from pseudorandom state generators (PRSGs)~\cite{C:JiLiuSon18}, a quantum analog of PRGs. While PRSGs can be constructed from OWFs, they are unlikely to imply OWFs~\cite{Kre21,STOC:KQST23,STOC:LomMaWri24}. We use the term \emph{Microcrypt} to denote the set of cryptographic primitives that are potentially weaker than OWFs. 

Despite the large body of work in the subject, our understanding of the relation between different cryptographic primitives in Microcrypt is still limited. In particular, it is currently considered an open problem to determine what is the \emph{minimal} assumption in quantum cryptography.

\paragraph{EFI pairs  and one-way state generators.}
In this vein, Brakerski, Canetti, and Qian~\cite{ITCS:BraCanQia23} introduced the concept of EFI pairs, which are efficiently samplable pairs of quantum states that are statistically far but computationally indistinguishable. They showed that the existence of EFI pairs is equivalent to that of many quantum cryptographic primitives, such as commitments, zero-knowledge proofs, and multi-party computation. 
In a parallel line of work, Morimae and Yamakawa~\cite{C:MorYam22,TQC:MorYam24} proposed one-way state generators (OWSGs), which are a quantum analogue of OWFs, as another potential minimal assumption in quantum cryptography.
They demonstrated that the existence of OWSGs is implied by the existence of many primitives, including (pure-state) private-key quantum money schemes, 
SKE, and digital signatures. Roughly, a OWSG is a quantum polynomial-time (QPT) algorithm that maps a classical string $x$ to a (possibly mixed) quantum state $\psi_x$ satisfying the following requirements:\footnote{
OWSGs are first defined in ~\cite{C:MorYam22} for the case of pure state outputs and then generalized to the case of mixed state outputs in \cite{TQC:MorYam24}. 
In this work, OWSGs mean the mixed state output version, unless otherwise specified. 
} 
\begin{itemize}
\item {\bf Efficient verifiability:} There is a QPT verification algorithm $\V$ that accepts $(x,\psi_x)$ with an overwhelming probability over uniform $x$.
\item {\bf One-wayness:} Given polynomially many copies of $\psi_x$, no QPT adversary can find $x'$ such that $\V(x',\psi_x)$ accepts, except with negligible probability.
\end{itemize}
We can think of EFI pairs and OWSGs as postulating two different types of hardness, that in the classical world are equivalent~\cite{Gol90}: EFI pairs postulate hardness of decision, namely that of a distinguishing problem. On the other hand, OWSGs postulate hardness of search, namely that of inverting an easy-to-compute function, where solutions are efficiently verifiable. Classically, these two hardness notions are equivalent due to the Goldreich-Levin theorem \cite{GL89}. Furthermore, in the classical world the existence of an efficient verification algorithm is without loss of generality, since one can verify the validity of a pair $(x, f(x))$, by simply recomputing $f$.

A recent breakthrough by Khurana and Tomer~\cite{STOC:KhuTom24} connected these two primitives by showing that OWSGs with pure state outputs imply EFI pairs. Soon after that, Batra and Jain~\cite{BJ24} generalized their result to the case of mixed state outputs. In fact, they show equivalence between  EFI pairs and inefficiently-verifiable OWSGs (IV-OWSGs), which are a weaker variant of OWSGs where the verification algorithm is allowed to run in unbounded-time.\footnote{
Batra and Jain~\cite{BJ24} refer to IV-OWSGs as statistically-verifiable OWSGs (sv-OWSGs).
The term ``IV-OWSG'' was introduced by 
Malavolta, Morimae, Walter, and Yamakawa \cite{MMWY24} who  concurrently showed a similar result but with an exponential security loss.} 
As a consequence, we now understand that most quantum cryptographic primitives imply EFI pairs (or, equivalently, IV-OWSGs) suggesting that they currently represent the minimal assumptions in quantum cryptography. However, at present, nothing is known about the \emph{reverse direction}: Could it be that EFI pairs imply OWSGs, or is it the case that EFI pairs characterize a separate world in the depths of Microcrypt?
Note that, by \cite{BJ24}, this is equivalent to asking whether IV-OWSGs imply OWSGs. In other words:
\begin{quote}
 \centering
 \emph{Does efficient verification come ``for free'' in quantum cryptography?}
\end{quote}
The objective of our work is to make progress on this question. As we shall discuss next, the phenomenon of efficient vs inefficient verification is quite common in quantum cryptography, and many of the known open problems revolve around this question.


\paragraph{The QCCC model.} A well-studied model in the literature is the setting of quantum computation with classical communication (QCCC), where quantum computations are performed locally, but all communication is restricted to be classical. In this model, \cite{STOC:KhuTom24} postulated the existence of a primitive, called one-way puzzles (OWPuzzs), which directly generalizes OWFs. A OWPuzz consists of a QPT sampling algorithm $\S$ and an unbounded-time verification algorithm $\V$, where $\S$ produces a pair of classical strings: a key $k$ and a puzzle $s$ such that $\V(k,s)$ accepts. The one-wayness requires that no QPT adversary, given $s$, can find a key $k'$ such that $\V(k',s)$ accepts, except with negligible probability. In \cite{STOC:KhuTom24} it is shown that OWPuzzs are implied by many primitives in the QCCC model, including public and symmetric key encryption, digital signatures, and commitments. 

This idea was further developed in \cite{C:ChuGolGra24}, who introduced an efficiently verifiable variant of OWPuzzs (EV-OWPuzzs) 
where $\V$ is a QPT algorithm, and showed that many (but not all) of the aforementioned primitives in the QCCC model also imply EV-OWPuzzs. For instance, the QCCC version of EFI pairs ($\qefid$ pairs), where the distributions are classical bits, are not known to imply EV-OWPuzzs. This question was left open in \cite{C:ChuGolGra24}. 

The techniques developed in \cite{C:ChuGolGra24}, which are in turn based on \cite{Kre21}, crucially rely on the fact that EV-OWPuzzs capture the hardness of problems with classical inputs and classical outputs, but do not seem to extend to the context of \emph{quantum} inputs (and classical outputs), such as $\qefid$ pairs.

\paragraph{Unclonability.} A defining property of quantum mechanics is that of \emph{unclonability of quantum states}. In cryptography, this property is used in the context of (private-key) quantum money schemes \cite{Wiesner83,C:JiLiuSon18}. A private-key quantum money scheme allows one to mint banknotes $\$_k$ in the form of quantum state, and it should be hard to clone banknotes, i.e., to create more than $t$ valid banknotes, given $\$_k^{\otimes t}$. Importantly, given the secret key $k$, one can efficiently verify the validity of a banknote state.

Although quantum money schemes are a central cryptographic primitive, that arguably started the field of quantum cryptography, very little is known on the relation with other Microcrypt primitives. For the special case of pure money states, it is known that quantum money schemes imply OWSGs~\cite{TQC:MorYam24}, but no general relation is known to hold either way.

\begin{figure}
    \centering
\begin{tikzpicture}[
  node distance=2cm and 3cm,
  primitive/.style={draw, rounded corners, align=center, font=\small},
  ref/.style={font=\scriptsize, midway, sloped, above},
  ref-nonsloped/.style={font=\scriptsize},
  scale = 0.8
]

\node[primitive] (unclstates) at (0,2) {UCSG};
\node[primitive] (pureOWSG) at (0,0) {Pure OWSG};
\node[primitive] (OWSG) at (0,-1.5) {OWSG};
\node[primitive] (QMoney) at (3,-1.2) {QMoney};
\node[primitive] (IVOWSG) at (0,-5) {IV-OWSG};
\node[primitive] (EVOWPuzz) at (7,0) {EV-OWPuzz};
\node[primitive] (OWPuzz) at (7,-4) {OWPuzz};
\node[primitive] (QEFID) at (3,-3) {QEFID};
\node[primitive] (nuQEFID) at (3,-4) {nuQEFID};
\node[primitive] (EFI) at (3,-5) {EFI};

\node[font=\small] at (-3,-1) {\bf Efficiently Verifiable};
\node[font=\small] at (-3,-3.5) {\bf Inefficiently Verifiable};

\draw[color=red, ->] (unclstates) -- (pureOWSG) node[ref-nonsloped, right, midway] {\begin{tabular}{c} Thm. \\ \ref{thm:owsg-unclonable-states} \end{tabular}};
\draw[->] (pureOWSG) -- (OWSG);
\draw[->] (pureOWSG) to [bend left=30] node[ref] {\cite{STOC:KhuTom24}} (OWPuzz);
\draw[->] (OWSG) -- (IVOWSG);
\draw[->] (EVOWPuzz) to [bend right=20] node[ref] {\cite{C:ChuGolGra24}} (OWSG);
\draw[->] ($(EVOWPuzz.south) + (-0.3,0)$) -- ($(OWPuzz.north) + (-0.3,0)$);
\draw[->] (QEFID) -- (nuQEFID);
\draw[->] (nuQEFID) -- (EFI);
\draw[<->] (OWPuzz) -- (nuQEFID) node[ref] {\cite{STOC:KhuTom24,C:ChuGolGra24}};
\draw[<->] (IVOWSG) -- (EFI) node[ref] {\cite{BJ24}};
\draw[color=red, ->] (QEFID) to node [sloped] {\(\not\)}node[ref,above=1mm] {Thm. \ref{thm:efid-vs-owsg}} (OWSG);
\draw[color=red, ->] (QEFID) to node [sloped] {\(\not\)}node[ref-nonsloped, right] {\begin{tabular}{c} Thm. \\ \ref{thm:efid-vs-money} \end{tabular}} (QMoney);
\draw[->] ($(OWPuzz.north) + (0.3,0)$) to node [sloped] {\(\not\)}
node[ref,above=0.5mm]{\cite{C:ChuGolGra24}}
($(EVOWPuzz.south) + (0.3,0)$);

\draw[dotted] (-4.6,-2.5) -- (8,-2.5);

\end{tikzpicture}
\caption{Implications among Microcrypt primitives, where pure OWSG means OWSGs with pure state outputs and nu{\qefid} means non-uniform {\qefid} pairs. We regard EFI pairs and (nu){\qefid} pairs as inefficiently verifiable due to their equivalence (up to non-uniformity) to IV-OWSGs and OWPuzzs, respectively. 
Arrows without references indicate trivial implications.
\label{fig:eff-ineff-onewayness+main_results}}
\end{figure}
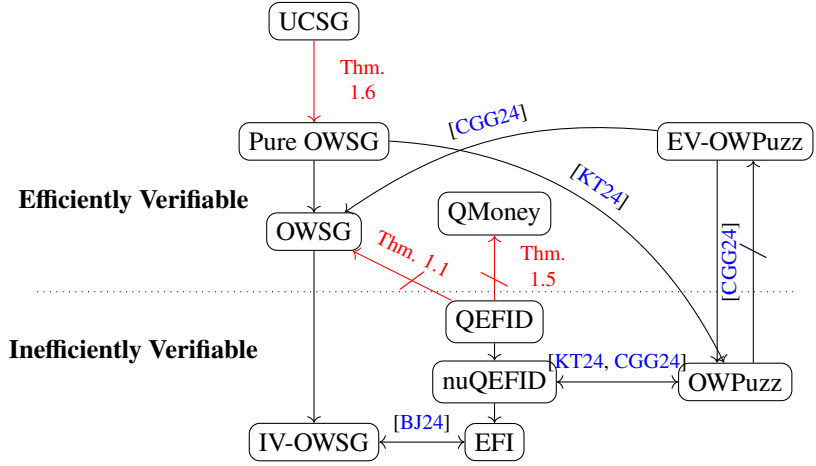

\subsection{Our Results}
In this paper, we make progress on the above questions, and we show strong evidence that there is a qualitative difference between efficient verification and one-wayness in quantum cryptography. We provide a schematic overview of our results in \Cref{fig:eff-ineff-onewayness+main_results}, and we discuss our main theorems in the following.

Our main result shows a separation between OWSGs and EFI pairs (\Cref{thrm:efi-exist,thm:break-OWSG}). In fact, our main theorem is slightly more general.

\begin{theorem}[OWSGs vs $\qefid$ Pairs]\label{thm:efid-vs-owsg}
    There exists a unitary oracle relative to which exponentially-hard $\qefid$ pairs exist but OWSGs do not.
\end{theorem}
$\qefid$ pairs are a variant of EFI pairs where the outputs are not quantum states but classical bit strings (note that the generation algorithm is still QPT). Since $\qefid$ pairs trivially imply EFI pairs, we have the following corollary, which solves the first open question negatively.
\begin{corollary}[OWSGs vs EFI pairs]
    There exists a unitary oracle relative to which exponentially-hard EFI pairs exist but OWSGs do not.
\end{corollary}
Moreover, our main theorem leads to further implications, which we hereby summarize \ifnum\submission=1
(details are given in \Cref{sec:further-implications}).
\else
.
\fi

Since EV-OWPuzzs imply OWSGs \cite[Theorem 32, Section 10]{C:ChuGolGra24}, 
we obtain a separation between the analogous primitives of EFI pairs and OWSGs, in the QCCC model. This gives a negative answer to a possibility raised in \cite{C:ChuGolGra24}.
\begin{corollary}[EV-OWPuzzs vs $\qefid$ pairs]\label{cor:EFID-EV-OWPuzz}
    There exists a unitary oracle relative to which $\qefid$ pairs exist, but EV-OWPuzzs do not.
\end{corollary}
Finally, since $\qefid$ pairs imply OWPuzzs~\cite[Lemma 8]{C:ChuGolGra24}, our main theorem also separates OWSGs from OWPuzzs. 
\begin{corollary}[OWPuzzs vs OWSGs]\label{cor:owpuzz-owsg}
    There exists a unitary oracle relative to which OWPuzzs exist, but OWSGs do not.
\end{corollary}
Overall, our results provide a clearer picture of quantum one-wayness, by showing a separation between efficiently and inefficiently verifiable versions of the different cryptographic primitives.

Next, we turn our attention to the relation between unclonable cryptography and other primitives in Microcrypt. As our main theorem, we show the first separation between quantum money and $\qefid$ pairs.
\begin{theorem}[Quantum Money vs {\qefid}]\label{thm:efid-vs-money}    
There exists no fully black-box construction of private-key quantum money schemes from exponentially-hard \qefid pairs.
\end{theorem}
On the other hand, using similar techniques devised to prove \Cref{thm:efid-vs-owsg}, we show that 
unclonable state generators (UCSGs) imply OWSGs. UCSGs (implicilty defined in \cite{C:JiLiuSon18}) have the same syntax as quantum money, but they satisfy a slightly stronger unclonability guarantee.
\begin{theorem}[UCSGs vs OWSGs]\label{thm:owsg-unclonable-states}
    UCSGs imply pure OWSGs.
\end{theorem}
Finally, combining \Cref{thm:efid-vs-owsg,thm:owsg-unclonable-states}, we obtain the following implication.
\begin{corollary}[UCSGs vs $\qefid$ pairs]\label{cor:uncl-states-efid}
    There exists a unitary oracle relative to which $\qefid$ pairs exist, but UCSGs do not.
\end{corollary}

\subsection{Overview of Techniques}\label{sec:techniques}

To obtain our main result, namely, the separation of OWSGs from \qefid pairs, we aim to find a unitary oracle that provides us with the right amount of ``hardness'' on which we may base the construction of an \qefid pair yet, at the same time, gives enough ``computational power'' to break OWSGs, deeming them non-existent. As is often the case in oracle separations, it is thus easier to split our task and regard it as building two oracles: one that models the hardness sufficient for \qefid pairs and another that models a powerful inverter that breaks OWSGs.

Before describing our oracles, let us give some underlying intuition. Statistical-vs-computation gap between distributions (as implied by {\qefid} pairs) necessitates producing randomness. In the classical setting, where computation is always deterministic, such randomness may be generated only by applying a deterministic procedure over a random seed. Such a procedure is \emph{derandomizable}. Additionally, one can efficiently verify whether an output was generated using a given seed, by repeating the computation. Our construction of \qefid pairs cannot be derandomizable, i.e., it cannot possibly rely solely on such a source of randomness since it cannot be secure in a world where any OWSG is broken.\footnote{Such a construction of \qefid pairs, by the same proof as in~\cite{Gol90}, would imply (quantum-evaluatable) PRGs that trivially imply OWSGs.} Our \qefid pairs therefore must attain their entropy from a different source of randomness, namely from quantum uncertainty. In the quantum world, we are able to produce randomness through non-deterministic quantum computation that involves measurements. Importantly, such a procedure is \emph{underandomizable} and does not generally induce efficient verifiability of its outcome.

\paragraph{Building \qefid pairs relative to an oracle.}\label{pg:oracle} A most basic underandomizable quantum procedure that generates a random outcome is the following: prepare a uniform superposition $\sum_{x\in X}\ket{x}$ then measure in the standard basis. This will sample a uniformly random element in $X$. A natural way to produce a pseudorandom distribution over $\sec$ bits (which implies, in particular, {\qefid} pairs), then, is to somehow create the state $\sum_{s\in S}\ket{s}$, where $S$ is a large enough uniformly random subset of $\bin^\sec$, then measure to obtain a uniformly random $s\gets S$. This brings us to define our first oracle: The oracle $O$ is a controlled swap between $\ket{0^\sec}$ and $\ket{S}=\sum_{s\in S}\ket{s}$, where $S$ is a uniformly random subset of size $2^{\sec/2}$ over $\bin^\sec\setminus\{0^\sec\}$. That is, given control bit 1, the oracle maps $\ket{0^\sec}$ to $\ket{S}$, $\ket{S}$ to $\ket{0^\sec}$, and acts as identity on the subspace orthogonal to the span of $\ket{0^\sec}$ and $\ket{S}$. On control bit 0, the oracle is identity. Our \qefid pairs under $O$ simply call the oracle $O$ with input $\ket{1}\ket{0^\sec}$ to obtain $\ket{S}=\sum_{s\in S}\ket{s}$, then measure to get a uniformly random element from $S$.\footnote{We could potentially define our oracle to perform the measurement as well, outputting a random $s\gets S$ at every query. This would not be, however, a unitary oracle as we have promised.}

To argue why our \qefid pairs give a pseudorandom distribution under the oracle $O$ (remember there is one more oracle that we ought to introduce), we apply a \emph{reflection emulation} technique, first introduced by Ji, Liu and Song~\cite{C:JiLiuSon18}. In their work, they show that access to a reflection oracle about a state $\ket{\psi}$, i.e. the unitary oracle $\reflect_\psi=I-2\proj{\psi}$, may be simulated given sufficiently many copies of $\ket{\psi}$. This is useful to us due to the key observation that the oracle $O$ is, in fact, the reflection unitary about the state $\ket{1}\ket{S-}$, where $\ket{S-}=\frac{1}{\sqrt{2}}(\ket{S}-\ket{0^\secp})$. Consequently, any algorithm $\adv$ with access to $O$ that is successful in breaking our \qefid pairs may be simulated by an algorithm $\badv$ that is not given access to $O$, but takes as input sufficiently many copies of $\ket{S-}$, and breaks \qefid pairs as successfully.

Therefore, we have reduced our task to showing that any such $\badv$ is incompetent in breaking the \qefid construction. Recall such $\badv$ takes as input many copies of the state $\ket{S-}$, for a uniformly random subset $S$ 
of size $2^{\sec/2}$, and a string $s\in\bin^\sec$, which is sampled either uniformly at random or uniformly at random from the subset $S$. To show indistinguishability between these two cases, we bound the trace distance between the mixed states corresponding to the two distributions. Importantly, this argument is statistical and holds even given any additional oracle that is independent {of} $O$.

\paragraph{Breaking OWSGs.} Having shown the existence of \qefid pairs relative to the oracle $O$, it remains to prove the other half of our statement: OWSGs do not exist. This requires complementing the oracle $O$ with an additional oracle that gives the necessary computational power for such a task.

The question of identifying an oracle under which there is no efficiently verifiable quantum one-wayness has been studied in the literature. Cavalar et al.~\cite{cavalar2023computational} show that, given quantum access to a \PP oracle, there exists a generic attack that breaks any pure OWSGs. Besides the fact that it is not clear how to extend their attack to any OWSGs, there is a more inherent limitation to adapting their approach to our setting. More specifically, their attack breaks OWPuzzs, which in turn implies that pure OWSGs are broken (due to~\cite{STOC:KhuTom24}, see \cref{fig:eff-ineff-onewayness+main_results}). Since \qefid pairs too imply OWPuzzs, any attempt to apply their attack (or any strengthening thereof) to our world, consisting of the oracle $O$, will result in an attack against our \qefid construction if successful. Therefore, we cannot hope to use their approach to establish a separation by, for instance, appending the oracle $\PP^O$ to our world.

A different method to generically break efficiently verifiable one-wayness is proposed by the framework of \emph{shadow tomography}~\cite{Shadow2,NatPhys:HKP20}, specifically the task of \emph{gentle search}~\cite{Shadow2}. A gentle search algorithm takes, as input, a collection of POVM elements $\{\Pi_k\}_{k\in\keyspace}$ (indexed by an arbitrary key space $\keyspace$) and a state $\ket{\phi}$. Its goal is to output a key $k\in\keyspace$ such that $\Pi_k$ accepts $\ket{\phi}$ with good probability, i.e., $\Tr(\Pi_k\ketbra{\phi}{\phi})$ is large, assuming such a $k$ exists. 
Breaking OWSGs via gentle search is straightforward: Letting the POVM $\Pi_k$ denote verification under key $k$, gentle search finds a key under which the input state is verified, hence successfully inverting OWSGs. As observed by~\cite{CCS24}, the gentle search procedure from~\cite{Shadow2} can be performed by a \qpspace computation\footnote{Note that this \qpspace is not the class of classical decision problems that can be decided by quantum polynomial-space computing. This is a unitary operation over polynomial number of qubits.} when the POVMs are implemented by polynomial-size quantum circuits, even when these circuits themselves have access to a \qpspace oracle. Consequently, by augmenting our world with a \qpspace oracle,\footnote{As we have mentioned, this is not an oracle that solves classical decision problems, but a unitary operation. Therefore, this oracle takes quantum states as input and outputs quantum states.} we may break any OWSG construction that uses the \qpspace oracle. This is, however, not entirely sufficient for our separation result to hold: We must rule out any OWSG construction that possibly uses the oracle $O$ as well.

A straightforward attempt would be to equip the attack with the oracle $\qpspace^O$, i.e., now the oracle can implement any poly-space quantum circuit that has access to $O$. Such an oracle, however, is too strong as it breaks our \qefid construction as well due to~\cite{cavalar2023computational} (since the attack using the $\PP$ oracle can be performed using the $\qpspace$ oracle). Instead, we observe that the reflection emulation technique, which we used previously to prove the security of our \qefid pairs, comes in handy here as well. By reflection emulation, we can replace the verification circuit of OWSGs, which possibly makes queries to $O$ and $\qpspace$, with a circuit that makes queries only to $\qpspace$, yet is given in addition copies of the state $\ket{S-}$. We then apply the gentle search procedure described above w.r.t. the POVM elements $\{\Pi_k\}_k$, where $\Pi_k$ implements the new $\qpspace$-aided verification circuit with key $k$.

Of course, now the POVM elements expect as input, besides the state $\ket{\phi}$ to be verified, a polynomial amount of copies of $\ket{S-}$. Hence, we must provide the gentle search with sufficiently many such copies. To generate these copies, we apply phase estimation for the unitary $\reflect_{S-}:=I-2\ketbra{S-}{S-}$ on the state $\ket{0^\sec}$ using $O$: Recall $O$ is a reflection about $\ket{1}\ket{S-}$ and hence a controlled-version of $\reflect_{S-}$, and $\ket{0^\sec}$ is a uniform superposition of $\ket{S-}=\frac{1}{\sqrt{2}}(\ket{0^\secp}-\ket{S})$, which has eigenvalue $-1$ under $\reflect_{S-}$, and $\ket{S+}= \frac{1}{\sqrt{2}}(\ket{0^\secp}+\ket{S})$, with eigenvalue $1$. 

To conclude, our attack against any OWSG construction with efficient verification circuit $\V$, under oracles $O$ and \qpspace, takes as input sufficiently many copies of the {challenge} state $\ket{\phi}$ to be inverted, and performs the following steps: \begin{enumerate*}
    \item Compute the \qpspace-aided circuit $\V'$, which emulates $\V$ while replacing the calls to the oracle $O$ with 
    taking as input copies of the state $\ket{S-}$.
    \item Generate sufficiently many copies of $\ket{S-}$ via phase estimation for $O$.
    \item Using the \qpspace oracle, perform gentle search w.r.t. the POVM elements $\{\Pi_k\}_k$, where $\Pi_k$ implements $\V'$ under key $k$, and input constituting of $\ket{\phi}$ and copies of $\ket{S-}$ (at this point, we have sufficiently many copies of this input as required by gentle search).
\end{enumerate*}

The attack against any OWSGs under the oracles $O$ and \qpspace, together with the existence of a \qefid pair under these oracles, completes the oracle separation between OWSGs and \qefid pairs.

We now proceed to discuss the techniques underlying our second main result, namely separating private-key quantum money schemes from \qefid pairs.

\paragraph{Separating Private-key Quantum Money schemes from \qefid pairs.} Recall that a private-key quantum money scheme consists of three QPT algorithms, $\kgen$ that samples a key $k$, 
$\mint$ that mints a quantum money state $\$_k$ (which can be generally mixed), and $\verify$ that on receiving a key $k$ and an alleged money state $\rho$ either accepts or rejects. 
The money scheme is called unforgeable if it is hard, for any QPT adversary who is given money states minted under $k^*\gets\kgen$, to produce new money states that are valid under the same key $k^*$. That is, for any polynomial $m$,
the adversary is given $\$_{k^*}^{\otimes m}$ sampled by running $\mint(k^*)$ for $m$ times, where $k^*\gets\kgen$. 
The adversary is given access to the verification oracle w.r.t. $k^*$ and it breaks the scheme if it outputs a quantum state
that contains at least $m+1$ valid money states that pass verification under $k^*$.

When quantum money schemes produce pure money states 
they imply OWSGs~\cite{TQC:MorYam24}. 
However, for general mixed money states, we do not know whether private-key quantum money schemes imply OWSGs or not.
This gap becomes evident when we try to break private-key quantum money schemes
by using the same attack as the one used to break OWSGs in the proof of \Cref{thm:efid-vs-owsg}.

Indeed, assume we try to apply gentle search to find the key $k^*$ using which the given money states were minted. Clearly, given such a key we may reproduce as many valid money states as we wish. The only guarantee given by gentle search, however, is that the key it finds accepts the input money state $\$_{k^*}$ with good probability.  That is, it is possible that the gentle search returns a key $k\neq k^*$ such that $\$_{k^*}$ is valid also under $k$, yet producing a new money state using $k$ will result in a money state $\$_k$ that does not pass verification under $k^*$. For example, consider a quantum money scheme where there is a dummy key $k_0$ under which any state passes verification. Gentle search, in this case, might always return $k_0$ and is hence useless.\footnote{This being said, we note that, when the money states are pure, there is a successful alternative way to apply gentle search. Rather than considering the POVMs that implement the verification circuit under $k$, perform gentle search w.r.t. the POVMs that implement a swap test between their input and the pure money state minted under the corresponding key $k$ (essentially, performing a projection into the latter).}

A different strategy is then required to attack arbitrary quantum money schemes that make use of the oracle $O$.
While several such statistical attacks have been proposed in the literature~\cite{Aar16,Shadow2}, none seem to be applicable in our oracle-relative world: An extension of the gentle search procedure~\cite{Aar16} which succeeds in breaking any private-key quantum money scheme in the plain model requires calling the verification circuit an exponential number of times, consequently requiring exponential number of calls to the oracle $O$ in our setting (or preparing exponentially many copies of $\ket{S-}$ if we wish to evoke reflection emulation). Another known attack, in some ways more efficient, is based on the stronger tool of \emph{shadow tomography}~\cite{Shadow2}. Unfortunately, as we discuss next, the attack from~\cite{Shadow2} also fails in our setting, but we manage to use shadow tomography in a different attack to break any private-key quantum money scheme relative to $O$.

Shadow tomography may be seen as a strengthening of gentle search. In shadow tomography, we are again given a collection of POVM elements $\{\Pi_k\}_{k\in\keyspace}$ and an input state $\rho$, and the goal is to output, up to some precision error, the acceptance probability of $\rho$
under \emph{all given POVM elements}, i.e. $\Tr(\Pi_k\rho)$ for all $k$ (recall gentle search outputs a key where such probability high enough). As noted by~\cite{Shadow2}, shadow tomography immediately gives an attack against any private-key quantum money scheme in the plain model. Letting the POVM element $\Pi_k$ correspond to the acceptance of the verification circuit under key $k$, shadow tomography allows us to (approximately) 
estimate the probability that a given state $\$$ passes verification under each key,
namely, $\Pr[\top\gets\verify(k,\$)]$ for each $k$. Then, given $\$_{k^*}\gets\mint(k^*)$, an attacker can find, by brute-force, a different state $\$'$ that has close-enough acceptance probabilities under all keys and output it. The attacker can produce as many copies of $\$'$ as he wants. In particular, since $\$_{k^*}$ passes verification under $k^*$ with good probability, then so does $\$'$.

The above approach completely fails when verification has access to the oracle $O$, even with reflection emulation in hand. Consider carrying the above attack w.r.t. a verification circuit that, instead of querying $O$, takes as input polynomially many copies of $\ket{S-}$ besides $\$$ (as obtained by the reflection emulation technique from~\cite{C:JiLiuSon18}). Via brute-force, we are able to find all states that exhibit acceptance probabilities similar to the input $\$\otimes\ket{S-}^{\otimes t}$ (where $t$ is the number of copies required by reflection emulation). These may include states of the form $\$'\otimes \ket{S'-}^{\otimes t}$, for $S'\neq S$. Since the number of possible subsets $S'\subset \bin^\secp\setminus\{0^\secp\}$ of size $2^{\secp/2}$ is doubly-exponential, it is impossible to identify those satisfying $S=S'$ using polynomially many queries to the oracle $O$ (as this would imply an efficient algorithm for learning $S$ and consequently breaking our \qefid pairs). 

Alternatively, one may attempt to perform the brute force over the money state register alone; to any possible value it may take, attach the copies $\ket{S-}^{\otimes t}$ that we take as input, then invoke shadow tomography thereon. For every such possible money state, we obtain the acceptance probabilities under all keys, which we compare to those of the given money states until a match is found. Every invocation of shadow tomography, however, disturbs the copies of $\ket{S-}$. Such disturbance, even if in the slightest, limits the number of states we can brute-force over to an insufficient number of trials.

While the above brute-force idea does not work for us, it serves as inspiration to our ultimate solution. We observe that we need not brute-force over all possible states that the money state register may take. Rather, it is sufficient to look at states that can be produced by the minting algorithm of the private-key money scheme, i.e., $\$_k\gets\mint(k)$ for any key $k$. Roughly speaking, since these states have classical description (i.e. the corresponding keys), we are able to ``brute-force'' over them via a single invocation of shadow tomography, thus requiring only a polynomial number of copies of $\ket{S-}$! In more details, just as before, we apply shadow tomography to estimate the probability of acceptance of the input state $\$_{k^*}$ under any key $k$, which we denote by $b^*_k$. Next, we apply shadow tomography w.r.t. the POVM elements $\{(\Pi_{k',k},I-\Pi_{k',k})\}_{k,k'}$ and the all-zero state, where, for any pair of keys $k',k$, $\Pi_{k',k}$ corresponds to the following event:
\begin{enumerate*}
    \item Generate $\$_{k'}\gets \mint(k')$ by running $\mint(k')$ (which can be implemented as a unitary since $k'$ is classical) on many copies of the state $\ket{S-}^{\otimes t}$.
    \item Run $c\gets \verify(k,\$_{k'})$ and the verification accepts.
\end{enumerate*}

For any key $k'$, we obtain the acceptance probability of $\$_{k'}$ under all possible keys $k$, which we denote by $\{b_{k',k}\}$. We choose a key $k'$ such that $\{b_{k',k}\}_k\approx\{b^*_{k}\}_k$, and output sufficiently many states minted using $k'$. Note that such a $k'$ exists since, for $k'=k^*$, we have $\{b_{k',k}\}_k\approx\{b^*_{k}\}_k$ as $\$_{k'}$ distributes like the input $\$_{k^*}$. In fact, it is evident that any key $k'$ satisfying $\{b_{k',k}\}_k\approx\{b^*_{k}\}_k$ will be good enough for forfeiting purpose as money minted using $k'$ will have acceptance probability under $k^*$ close to that of $\$_{k^*}$ (which is close to $1$ by correctness), and hence the attack is successful. 

We note that our attack against private-key quantum money schemes is an ($O$-aided) statistical attack with inefficient runtime and space complexity. However, it is query-efficient. It is not clear whether there exists an oracle $O'$ under which such an attack can be made efficient and still be able to break any private-key quantum money scheme, possibly using the oracle $O'$. Thus, in contrast to the first separation result of OWSGs from \qefid pairs, our separation of private-key quantum money schemes from \qefid pairs is, formally speaking, a fully-black-box separation and does not satisfy the stronger notion of an oracle separation between the two primitives. We leave the question of establishing an oracle separation between private-key quantum money schemes and \qefid pairs to future work.

\subsection{On the Worst-Case Simulatability of Our Oracles}
Exisiting separations involving Microcrypt primitives~\cite{Kre21,CCS24,CS24,cryptoeprint:2024/1043}
heavily rely on the concentration property of the Haar measure, such as Lévy’s lemma~\cite[Theorem 7.37]{Wat18}.
In contrast, the oracles with respect to which we achieve \Cref{thm:efid-vs-owsg,thm:efid-vs-money} are reflections about random subset states.
 This is the first example of separations between Microcrypt primitives in an oracular world that does not involve common Haar random states or unitaries.
We elaborate on what we think is a feature of our oracles which, despite seemingly technical, we believe offers the potential for further applications of our approach to new separation results.  
 
In both our separation and separations relative to common Haar random states or unitaries, the proofs involve ``de-oraclizing'' certain algorithms, i.e., simulating oracle-aided algorithms by algorithms that do not query any oracles, albeit via completely different and incomparable simulation techniques.

In our proof, we invoke a simulation technique from the work of~\cite{C:JiLiuSon18}, which we refer to as \emph{reflection emulation} (see \Cref{sec:techniques}). Via reflection emulation we are able to simulate our reflection oracle \emph{in the worst-case}, independently of the distribution on the oracle. In sharp contrast, the simulation of Haar random states or unitaries, as performed in previous works~\cite{Kre21,CCS24,CS24} crucially relies on the concentration property of the Haar measure, and hence the simulation guarantees are \emph{distribution-sensitive}, i.e., may not hold for a different non-Haar distribution over the same support.

This being said, due to the Haar concentration measure, the guarantee from Haar state simulation is stronger in terms of precision, as the concentration of Haar measure provides precision inverse-exponential in the dimension. In comparison, simulating reflection about a state $\ket{\psi}$ via reflection emulation provides precision inverse-polynomial in the number of copies of the state $\ket{\psi}$ that are given.
 
To demonstrate the advantage of worst-case simulation over a distribution-sensitive one, consider the following example. Given a Haar random unitary $U$ of dimension $D$, let us look at a QMA verifier $V^U$, that on any $n$-bit instance, takes as input an alleged witness state $\ket{\phi}\in \CC^D$, queries $U$ on $\ket{0...0}$ $n$ times to get the state $\ket{\psi_U}^{\otimes{n}}$, where $\ket{\psi_U}=U\ket{0...0}$, then performs a swap test on the given state $\ket{\phi}$ and $\ket{\psi_U}^{\otimes{n}}$. Clearly, if we apply standard Haar random emulation~\cite{Kre21}, then we would simulate queries to $U$ by sampling a different unitary $U'$ and running $V^{U'}$. In such a case, the largest possible acceptance probability of the obtained simulation $V^{U'}$, for any fixed oracle $U$, is $\frac{1}{2}+1/{\binom{2^n+D-1}{2^n-1}}$. In contrast, the original verifier $V^U$ has maximum acceptance probability $1$ as it always accepts the state $\ket{\psi_U}^{\otimes n}$. On the other hand, if we consider a reflection unitary $U$ and the same verifier as above, it can be easily shown that the worst-case reflection emulation results in the same maximum acceptance probability for the emulated verification, up to emulation precision.

\subsection{Related and Concurrent Work}\label{sec:comparison-to-concurrent works}

\paragraph{Related Work.} Separations inside Microcrypt have been studied before. \cite{C:ChuGolGra24} separates EV-OWPuzzs from PRSGs by ruling out EV-OWPuzzs in the oracular world of~\cite{Kre21}. Separations of QCCC key agreement protocols and QCCC commitments from Pseudorandom function-like state generators (PRFSGs) in the common Haar random state model were shown in~\cite{cryptoeprint:2024/1043}. Concurrent to the work of~\cite{cryptoeprint:2024/1043}, Chen, Coladangelo, and Sattath~\cite{CCS24} proposed the state oracle model to show separations and showed a separation between $1$-PRSGs and PRSGs in this model.

\paragraph{Concurrent Work.} In a later revision concurrent to our work,~\cite{CCS24} upgraded their separation to the unitary oracle model. In another concurrent work, \cite{BCN24} prove a similar result to \Cref{thm:efid-vs-owsg} by separating OWSGs from OWPuzzs and $1$-PRSGs. Common to our paper,~\cite{BCN24} and the later version of~\cite{CCS24} show separation in a unitary model between commitments (equivalently, EFI) and either OWSG or related primitives, such as PRSG.

Despite the similarities between our work and~\cite{BCN24,CCS24}, we outline a few technical differences between them.
The work of \cite{CCS24} rules out PRSGs while constructing $1$-PRSGs (and hence, quantum commitments), and~\cite{BCN24} rules out OWSGs while constructing $1$-PRSGs (and hence, quantum commitments) and OWPuzzs. On the other hand, in \Cref{thm:efid-vs-owsg}, we construct \qefid pairs (and hence OWPuzzs, and quantum commitments) while ruling out OWSGs in the respective (unitary) oracular worlds. 
    Notably, since one of the distributions in our construction of $\qefid$ pairs (see \Cref{con:efi}) is the uniform distribution, our works and~\cite{CCS24,BCN24} construct pseudorandom (but statistically far from the uniform) distributions but in different regimes. Namely, we construct a classical pseudorandom distribution but in a keyless regime, i.e., the quantum sampler simply samples a classical bit string, whereas \cite{CCS24,BCN24} construct $1$-PRSGs, which can be viewed as quantum pseudorandom distributions in a keyed regime, i.e., first, a classical key of size $\secp$ is sampled, and then a pure quantum state of size larger than $\secp$ is generated using the key. Note that it is not possible to construct a classical pseudorandom distribution in a keyed regime and also hope to rule out OWSGs, since such a keyed classical pseudorandom distribution will constitute a quantum-evaluatable PRGs, which trivially implies OWSGs.
    
    Second, our oracle differs from that of~\cite{CCS24,BCN24} as both works are based on the common Haar random state model and their unitarized version. In contrast, we consider an oracle corresponding to random subset states. Even though random subset states of appropriate size are known to be statistically close to Haar random states~\cite{JMW23,GB23}, this does not imply our oracle separation of OWSGs from \qefid pairs in \Cref{thm:efid-vs-owsg} also holds under the oracles considered in~\cite{CCS24,BCN24}, and vice-versa.
    This is because the statistical indistinguishability in~\cite{JMW23,GB23} is an average-case guarantee, which implies that the two oracle distributions are indistinguishable only on average. 
    Hence, even if there exists an oracle in the support of one distribution relative to which an adversary breaks every construction of the primitive (say OWSGs), there might not exist such an oracle and such an adversary in the second distribution.
    
\subsection{Open Questions}

Our work leaves the following questions open:
\begin{enumerate}
\item While our oracles do not involve sampling Haar random states and unitaries, our oracles are still quantum. The natural next step would be to upgrade the oracles to a classical oracle with quantum superposition access, which we leave open for future works.   
\item We show that there are no EV-OWPuzzs in the oracular world we consider, which by~\cite{C:ChuGolGra24} suggests that many quantum primitives in the QCCC model, including non-interactive QCCC commitments, cannot exist in this world. This still leaves open the question of whether QCCC commitments exist in this world. Any answer to this question would greatly affect our understanding of quantum commitments. A positive answer to this question would imply that the gap between OWSGs and commitments is independent of the quantum communication in the commitments protocol. On the other hand, a negative answer to the question would mean that QCCC commitments are strictly stronger than quantum commitments, implying that quantum communication is indeed a crucial resource for quantum commitments.
\item In the oracular world that we consider, \qefid pairs exist but OWSGs do not. Can we show an oracle separation in the opposite direction? That is, can we give an oracle relative to which OWSGs exist but \qefid pairs do not?
A slightly stronger question would be: can we give an oracle relative to which OWSGs exist but OWPuzzs (which are implied by \qefid pairs) do not? Note that the construction of OWPuzzs from OWSGs by~\cite{STOC:KhuTom24} crucially requires that the OWSGs are pure, and hence a positive answer to the question mentioned above does not contradict~\cite{STOC:KhuTom24}. Due to the construction of OWPuzzs from pure OWSG by~\cite{STOC:KhuTom24}, an oracle separation between arbitrary mixed OWSGs and OWPuzzs would imply a separation between pure and arbitrary mixed OWSGs.

\item We show a statistical attack on any private-key quantum money schemes in our oracular world. Can we improve it to an efficient attack, relative to some unitary oracle, thereby improving the separation between private-key quantum money schemes
and \qefid pairs to an oracle separation? 
\item In our oracular world, both EFI pairs and {\qefid} pairs exist: Is there an oracle separation between EFI pairs and \qefid pairs? What is the relation between the two notions? 
\end{enumerate}

\section{Preliminaries}
\paragraph{Notation and Terminology.}
We use standard notations of cryptography and quantum information.
We use $\secp$ and $\kappa$ to denote security parameters.
For a set $S$, we use $x\gets S$ to denote that $x$ is sampled uniformly at random from $S$. 
For a distribution $D$, we use $x\gets D$ to denote that $x$ is sampled according to the distribution $D$. 
For a distribution $D$, we use $x\in D$ to denote that $x$ is in the support of $D$. For two distributions $P=\{p_x\}_x$ and $Q=\{q_x\}_x$, we denote by $\SD(P,Q)\coloneqq \frac{1}{2}\sum_x|p_x-q_x|$ the statistical distance between $P$ and $Q$.
For two quantum states $\rho$ and $\sigma$, $\mathsf{TD}(\rho,\sigma)$ is their trace distance.

A quantum oracle is an oracle that realizes arbitrary (possibly inefficient) quantum channel. A unitary oracle is given by an oracle where the channel is a unitary. An oracle-aided (quantum) algorithm is a quantum algorithm that has access to a quantum oracle, i.e., the circuit description of the algorithm can have the oracle as a gate. We say that an algorithm is $O$-aided if it presumes access to the oracle $O$. We naturally extend this notation to a set $\oracle$, or distribution, of oracles when access to $O\in\oracle$ is given. We say that an oracle-aided algorithm is polynomial-query if there exist polynomials $q,p:\NN\to\NN$ such that the algorithm, on input of length $n$, makes at most $q(n)$ queries to its oracle, where each query is of length at most $p(n)$ (qu)bits.
$I$ is the two-dimentional identity operator. We often write $I^{\otimes n}$ as $I$ if the dimension is clear from the context.
For any pure quantum state $\ket{\phi}$, we denote the reflection unitary about $\phi$ by $\reflect_{\phi} = I-2\proj{\phi}$. We denote the controlled reflection unitary by $\reflect^c_{\phi}$. It holds that $\reflect^c_\phi = \reflect_{1,\phi} = I-2\ket{1}\proj{\phi}\bra{1}$.



\subsection{Quantum Cryptographic Primitives}

\ifnum\submission=1
We hereby provide formal definitions of the main quantum cryptographic notions considered in this work: \qefid pairs, OWSGs, and private-key quantum money schemes. Additional definitions of related notions are given in \Cref{sec:more-defs}.
\else
We hereby provide formal definitions of the main quantum cryptographic notions considered in this work: \qefid pairs, OWSGs, and private-key quantum money schemes. Additionally, we give definitions of related notions to which implications of our separation results extend.
\fi

We begin with the definition of \qefid pairs.

\begin{definition}[\qefid pairs]\label{def:efi}
Let $T:\NN\to\NN$ and $\epsilon:\NN\to[0,1]$. A \emph{$(T,\epsilon)$-\qefid} pair is a pair of distribution ensembles $(D_0,D_1):=(\{D_0(1^\sec)\}_\sec,\{D_1(1^\sec)\}_\sec)$ over classical bit strings that satisfy the following properties:
\begin{itemize}
    \item \textbf{Efficiently samplable:} There exists a QPT algorithm which, on input security parameter $1^\sec$ and a bit $b\in\bin$ samples from $D_b(1^\sec)$.
    \item \textbf{Statistically far:} There exists a polynomial $p:\NN\to\NN$ such that, for any $\sec\in\NN$, $$\SD(D_0(1^\sec),D_1(1^\sec))\ge 1/{p(\secp)}.$$
    \item \textbf{Computationally indistinguishable:} For any non-uniform quantum algorithm $\adv$ that, on input $1^\sec$ and a bit string $z$,
    runs in time $T(\sec)$, and any $\sec\in\NN$,
    $$\left|\Pr_{z\gets D_0(1^\secp)}[ \adv(1^\secp,z)=1]-\Pr_{z\gets D_1(1^\secp)}[ \adv(1^\secp,z)=1]\right|\le\epsilon(\secp).$$
\end{itemize}

We say that $(D_0,D_1)$ is a $(T,\epsilon)$-\qefid pair \emph{relative to an oracle $O$}, if the distributions are sampled by a polynomial-query $O$-aided algorithm and computational indistinguishability holds against any $O$-aided adversary $\adv$ that, on input $1^\sec$ and a bit string, runs in time at most $T(\sec)$.

Lastly, we say that such a pair is simply an \emph{\qefid} pair if $T$ is a polynomial and $\epsilon$ is negligible. We say that it is an \emph{exponentially-hard \qefid} pair if there exists a constant $c\in\NN$ such that $T(\sec)=\Omega(2^{\frac{\secp}{c}})$ and $\epsilon(\sec)=O(2^{-\frac{\secp}{c}})$. These notions immediately extend to the oracle-relative setting.

\end{definition}


Next, we recall the definition of one-way state generators (OWSGs) from~\cite{C:MorYam22}. 
Note that~\cite{BJ24} shows equivalence between {\it inefficiently-verifiable} OWSGs
and EFI pairs. Here, we consider {\it efficiently-verifiable} OWSGs.

\begin{definition}[OWSGs~\cite{TQC:MorYam24,C:MorYam22}]\label{def:owsg}
    A \emph{one-way state generator}, or \emph{OWSG} for short, with associated key space $\keyspace=\{\keyspace_\kappa\subseteq\bin^\kappa\}_\kappa$, 
    is a couple $(\G,\V)$ of QPT algorithms with the following syntax:
    \begin{itemize}
        \item $\phi\gets\G(k)$: On input a key $k\in\keyspace_\kappa$, the generation algorithm outputs a (possibly mixed) quantum state $\phi$,
        \item $\bin\gets\V(k,\phi)$: On input a key $k\in\keyspace_\kappa$ and a state $\phi$, the verification algorithm outputs 1 (accepts) or 0 (rejects),
    \end{itemize}
    and that satisfies the following properties:
    \begin{itemize}
        \item \textbf{Correctness:} There exists a negligible function\ifnum\submission=0\footnote{We note that this is equivalent to a definition with inverse poly correctness as correctness can be amplified by standard repetition.}\fi $\negl:\NN\to[0,1]$ such that 
        \[
            \Pr[1\gets \V(k,\phi);\ k\gets \keyspace_\kappa,\ \phi\gets\G(k)] \geq 1-\negl(\kappa).
        \]
        \item \textbf{One-wayness:} For any non-uniform QPT algorithm $\adv$ and any polynomial $t:\NN\to\NN$, there exists a negligible function $\negl:\NN\to[0,1]$ such that, for any $\kappa\in\NN$,
        \[
            \Pr[1\gets\V(k',\phi);\ k\gets \keyspace_\kappa,\ \phi\gets\G(k),\ k'\gets\adv(1^\kappa,\phi^{\otimes t(\kappa)})]<\epsilon(\kappa).
        \]
    \end{itemize}


    We say that $(\G,\V)$ is a OWSG \emph{relative to an oracle $O$} if the algorithms $\G$ and $\V$ are $O$-aided and satisfy correctness and one-wayness w.r.t any non-uniform polynomial-query $O$-aided quantum adversary $\adv$.
\end{definition}

In \Cref{sec:fbb-def}, we further define \emph{a black-box construction of private-key quantum money schemes from \qefid} pairs.

\ifnum\submission=0
We now give definitions for \emph{one-way puzzles} and the special case of \emph{efficiently verifiable one-way puzzles}.

\begin{definition}[OWPuzzs~\cite{STOC:KhuTom24} and Efficiently-Verifiable OWPuzzs~\cite{C:ChuGolGra24}]\label{def:EV-OWPuzz}
   A one-way puzzle (OWPuzz) is a pair $(\samp, \ver)$ of algorithms with $\samp$ being QPT, having the following syntax:
   \begin{enumerate}
       \item $(ans,puzz)\gets\samp(1^\kappa)$ takes as input the security parameter $\kappa$, and outputs two bit strings, an answer and a puzzle.
       \item $b\gets\ver(ans',puzz)$ takes an (alleged) answer, $ans'$, and the puzzle, and outputs 1 (accept) or 0 (reject).
   \end{enumerate}
We require the following two properties.
\begin{itemize}
\item \textbf{Correctness:} 
There exists a negligible function $\negl:\NN\to[0,1]$ such that
\begin{equation*} 
\Pr[1\gets\ver(ans,puzz); (ans,puzz)\gets \samp(1^\kappa)]\ge1-\negl(\kappa).
\end{equation*}
\item \textbf{Security:} For any non-uniform QPT algorithm $\adv$, 
there exists a negligible function $\negl:\NN\to[0,1]$ such that
\begin{align*} 
\Pr[1\gets \ver(ans',puzz); (ans,puzz)\gets \samp(1^\kappa), ans'\gets \adv(1^\kappa,puzz)]\leq\negl(\kappa).
\end{align*}
\end{itemize}
We say that $(\samp,\ver)$ is an efficiently-verifiable one-way puzzle (EV-OWPuzz) if $\ver$ is also a QPT algorithm.

\end{definition}

Lastly, we define \emph{unclonable state generators}, which were implicitly considered in the work of~\cite{C:JiLiuSon18}.

\begin{definition}[Unclonable State Generators (UCSGs) (\cite{uhlmann_complexity}, Implicit in~\cite{C:JiLiuSon18})]\label{def:unclonable}
    An unclonable state generator (UCSG) $\samp$ 
with a key space $\keyspace=\{\keyspace_\kappa\subseteq\bin^\kappa\}_\kappa$ 
    is a QPT algorithm that, on input a key $k$, outputs a pure state $|\phi_k\rangle$.
    We require the following property, which we call {\bf unclonability}:
        For any QPT adversary $\cA$,
        \begin{align}
            \Expct_{k\gets\keyspace_\kappa}\left[\Tr\left(\ketbra{\phi_k}{\phi_k}^{\otimes m+1}\cA(1^\kappa,|\phi_k\rangle^{\otimes m})\right)\right]\le \negl(\kappa).
        \end{align}

    We say that a UCSG is a \emph{strong} UCSG 
    if the security also holds against any adversary $\adv$ who has polynomial oracle access to the controlled reflection unitary $\reflect^c_{\phi_k}$.

    Lastly, we say that a UCSG is a (strong) UCSG \emph{relative to an oracle $O$}, 
    if $\samp$ is a polynomial-time $O$-aided generation algorithm, and satisfies the security w.r.t 
    any non-uniform polynomial-time $O$-aided (and, resp., $\reflect^c_{\phi_k}$-aided) quantum adversary $\adv$.
    
\end{definition}

\begin{remark}
    We note that despite the similarities between unclonability and unforgeability, there is a subtle difference in the definition of 
    UCSGs (\Cref{def:unclonable}) and private-key quantum money schemes (\Cref{def:quantum-money}) even when restricted to pure states. 
    In the cloning game (\Cref{def:unclonable}), the cloner can ask for any (polynomially large) $m$ copies of the state $\ket{\phi_k}$, and the winning condition can be interpreted as that the cloner wins if it 
    submits exactly $m+1$ registers, such that all the registers should pass the rank-$1$ projection into $\ket{\phi_k}$, whereas in the forging game (\Cref{def:quantum-money} restricted to pure states), 
    the forger after getting any $m$ copies of the money state $\ket{\phi_k}$ can output any (polynomial) number of registers and wins if least $m+1$ passes verification.
    This difference of passing at least $m+1$ verifications instead of submitting exactly $m+1$ registers and passing all the rank-$1$ projections can be bridged 
    if the cloning adversary gets access to the rank-$1$ projection into the state $\ket{\phi_k}$, and hence strong UCSGs and pure private-key quantum money schemes
    are equivalent.\footnote{Technically, the equivalence holds between strong UCSGs and pure private-key quantum money schemes with a rank-$1$ projection into the money state, 
    but the rank-$1$ verification is not an extra assumption, since any pure private-key quantum money scheme can be attached with a rank-$1$ verification.} 
\end{remark}
\fi

Lastly, we define private-key quantum money schemes.

\begin{definition}[Private-Key Quantum Money Schemes~\cite{C:JiLiuSon18}]\label{def:quantum-money}
    A private-key quantum money scheme is a tuple 
    $(\kgen,\mint,\verify)$ 
    of QPT algorithms 
    having the following syntax.
    \begin{enumerate}
    \item $k\gets \kgen(1^\kappa)$ takes a security parameter $\kappa$ and outputs a classical secret key, $k\in \bin^\kappa$. 
    \item  $\$_k\gets \mint(k)$  takes the secret key $k$ and outputs a quantum money state $\$_k$.
    \item $b\gets \verify(k,\$)$ receives the secret key $k$ and an (alleged) quantum money state $\$$, 
    and outputs 1 (accept) or 0 (reject).
\end{enumerate}
We require the following properties.
\begin{itemize}
\item \textbf{Correctness:} 
A private-key quantum money scheme is called $\mu$-correct, for $\mu:\NN\to[0,1]$, if for all $\kappa\in\NN$,
\begin{equation*}
\Pr[\verify(k,\$_k)=1; k\gets \kgen(1^\kappa), {\$}_k\gets \mint(k)]\geq \mu(\kappa).
\end{equation*}
We say that the scheme is simply \emph{correct} if $\mu$ is inverse-polynomial and we say it is \emph{perfectly correct} if $\mu=1$.
\item \textbf{Unforgeability:} For any polynomials $m,m':\NN\to\NN$ where $m'>m$, and for any non-uniform QPT algorithm $\adv$, there exists a negligible function $\negl:\NN\to[0,1]$ such that
\begin{align*} 
\Pr&[\exists S\subseteq [m'(\kappa)], |S|>m, \verify(k,\$_i)=1\forall i\in S;\\
&k\gets \kgen(1^\kappa), {\$}_k\gets \mint(k), \$_{1,\cdots,m'(\kappa)}\gets\adv^{\verify(k,\cdot)}(1^\kappa,\$_k^{\otimes m(\kappa)})]\leq\negl(\kappa),
\end{align*}
where $\$_{1,\cdots,m'(\kappa)}$ is a state on $m'(\kappa)$ registers and $\$_i$ denotes its $i^{th}$ register.

We say that the private-key quantum money scheme is \emph{statistically unforgeable} if unforgeability holds against any arbitrary (and not necessarily QPT) adversary that makes polynomially many queries to the verification oracle.

Further, we say that $(\kgen,\mint,\verify)$ is a private-key quantum money scheme \emph{relative to an oracle $O$} if the algorithms $\kgen$, $\mint$ and $\verify$ are $O$-aided and satisfy correctness and unforgeability w.r.t any non-uniform polynomial-time $O$-aided quantum adversary $\adv$. 
The oracle-relative notion extends to statistical unforgeability in the natural way.

\end{itemize}
\end{definition}

\subsection{On The Complexity of The Reflection Oracle}\label{sec:reflection}

The reflection unitary $R_{\phi}=I-2\proj{\phi}$ and its controlled variant $R^c_{\phi}=I-2\ket{1}\proj{\phi}\bra{1}$ play a significant role in our proofs. In particular, the latter is a component in the oracle under which we prove our separation results. In this section, we recall a couple of known facts regarding the power of the reflection oracle and its complexity.

We begin with the following folklore application of the controlled reflection oracle about a state $\ket{\phi}$ for projecting onto $\ket{\phi}$\ifnum\submission=1
, which we prove in \Cref{sec:reflection-proofs}
\else
.
\fi

\begin{restatable}[Controlled Reflection to Projection]{proposition}{reflecttoproject}
    \label{prop:project-via-reflect}
    Let $\ket{\phi}$ be a quantum state and let $R^c_{\phi}=I-2|1\rangle|\phi\rangle\langle\phi|\langle 1|$ be the corresponding controlled reflection unitary. There exists an $R^c_{\phi}$-aided quantum algorithm that takes as input a quantum state $\ket{\psi}$, makes a single query to its oracle, outputs $\ket{\phi}$ with probability $|\ip{\psi}{\phi}|^{2}$ and otherwise, with probability $1-|\ip{\psi}{\phi}|^{2}$, declares failure. 
\end{restatable}

\submitorfull{}{reflect-to-project}

Next, we recall the following theorem from~\cite{C:JiLiuSon18}, which states that the reflection oracle about $\ket{\psi}$ 
may be efficiently emulated given copies of the state $\ket{\psi}$.

\begin{restatable}[Reflection Emulation~{\cite[Theorem 4]{C:JiLiuSon18}}]{theorem}{reflectjls}
\label{thrm:reflect-jls}
    Let $Q$ be a quantum oracle. Let $\ket{\psi}$ be a quantum state and let $\reflect_{\psi}=I-2\proj{\psi}$ be the corresponding reflection unitary. Let $\ket{\phi}$ be a state not necessarily independent of $\ket{\psi}$. Let $\adv$ be a $(Q,\reflect_{\psi})$-aided quantum circuit that makes $q$ queries to the oracle $\reflect_{\psi}$. Then, there exists a $Q$-aided quantum circuit $\badv$ such that, for any $\ell\in\NN$,
    \[
        \TD\left(\adv^{Q,\reflect_{\psi}}(\ket{\phi})\otimes\ket{\psi}^{\otimes \ell}, \badv^Q(\ket{\phi}\otimes \ket{\psi}^{\otimes \ell})\right)\leq \frac{2q}{\sqrt{\ell+1}}.
    \]
  Further, if $\adv$ is of polynomial size then so is $\badv$.
\end{restatable}

The above theorem is, in fact, slightly stronger than Theorem 4 of \cite{C:JiLiuSon18} in the following two aspects.
First, we allow both $\adv$ and $\badv$ access to an additional oracle $Q$.
This is possible because in the proof of Theorem 4 of \cite{C:JiLiuSon18}, the emulation circuit $\badv$ uses $\adv$ as a black-box.
Second, our version of the theorem bounds the distance between the outcome states including the registers where the copies of $\ket{\psi}$ reside. That is, compared to Theorem 4 of \cite{C:JiLiuSon18}, we replace $\adv^{Q,\reflect_{\psi}}(\ket{\phi})$ with
        $\adv^{Q,\reflect_{\psi}}(\ket{\phi})\otimes\ket{\psi}^{\otimes \ell}$.
This is also possible, because, looking at the proof of Theorem 4 of \cite{C:JiLiuSon18},
$\badv$ does not disturb $|\psi\rangle^{\otimes \ell}$ when it uses them to emulate $R_\psi$. \ifnum\submission=0 We briefly recall the construction of $\badv$ from the proof by \cite{C:JiLiuSon18}.

Let $S=\symd$ denote the symmetric subspace over $\ell+1$ registers, where $\CC^N$ represents the Hilbert space in which the {$N$-dimensional} state $\ket{\psi}$ resides.\footnote{The symmetric subspace contains all states that are invariant to a permutation over their registers.} Let $R_S$ denote the reflection about the symmetric subspace. The circuit $\badv$ is defined the same as $\adv$ except that any query to the oracle $R_\psi$ is replaced by the application of $R_S$ over the input state (i.e. the register containing the query) and the $\ell$ registers containing the copies of $\ket{\psi}$.\fi

For completeness, we provide a proof of \cref{thrm:reflect-jls} in \ifnum\submission=1 \cref{sec:reflection-proofs}
\else
\cref{sec:jls-proof-full}
\fi
.

The following is an extension of \cref{thrm:reflect-jls} to the case where multiple reflection oracles are used. \ifnum\submission=1 A proof for this corollary too is attached in \Cref{sec:reflection-proofs}. \fi

\begin{restatable}{corollary}{reflectgeneral}
\label{cor:inductive-emulation}
    Let $Q$ be a quantum oracle. Let $\ket{\psi_1},\dots,\ket{\psi_m}$ be quantum states and let, for $j\in[m]$, $\reflect_{\psi_j}=I-2\proj{\psi_j}$ be the corresponding reflection unitary. Let $\ket{\phi}$ be a state not necessarily independent of $\ket{\psi_1},\dots,\ket{\psi_m}$. Let $\adv$ be a $(Q,\reflect_{\psi_1},\dots,\reflect_{\psi_m})$-aided quantum circuit that makes $q$ queries to its oracles $\{\reflect_{\psi_j}\}_j$ (in total). Then, there exists a $Q$-aided quantum circuit $\badv$ such that, for any $\ell\in\NN$,
    \begin{align*}
        \TD\left(\adv^{Q,\reflect_{\psi_1},\dots,\reflect_{\psi_m}}(\ket{\phi})\otimes\ket{\psi_1}^{\otimes \ell}\otimes\dots \otimes\ket{\psi_m}^{\otimes \ell},\right.& \\
        \badv^Q(\ket{\phi}\otimes\ket{\psi_1}^{\otimes \ell}\otimes\dots \otimes\ket{\psi_m}^{\otimes \ell})&\left.\vphantom{\badv^Q(\ket{\phi}\otimes\ket{\psi_1}^{\otimes \ell})}\right)\leq \frac{2q}{\sqrt{\ell+1}}.
    \end{align*}
    Further, if $\adv$ is of polynomial size then so is $\badv$.
\end{restatable}

\submitorfull{}{reflect-general}

\section{The Oracles}\label{sec:oracles}

We begin by defining the unitary oracle under which the separation between OWSGs and {\qefid} pairs from \cref{thm:efid-vs-owsg} holds. Our oracle is actually composed of two oracles: First, an oracle $O$, which captures the hardness required for building a \qefid pair. In fact, we define a distribution $\oracle$ over oracles, and we show that a ``hard'' oracle $O\in\oracle$ exists via a probabilistic argument. The distribution $\oracle$ will be useful also to establish the separation between private-key quantum money schemes and {\qefid} pairs (\cref{thm:efid-vs-money}). Second, an oracle $\qpspace$~\cite{CCS24}, which allows simulating any QPSPACE computation\footnote{Note that this QPSPACE computation does not mean the classical computation that can solve decision problems that are decided by quantum polynomial-space computing. It means a unitary operation over polynomial number of qubits. Therefore, our $\qpspace$ oracle is a quantum oracle that takes a quantum state as input and outputs a quantum state.} and provides the necessary power to break any OWSG candidate (while preserving the security of our {\qefid} pair).

We begin by defining a useful special case of the controlled reflection unitary.

\begin{definition}[Unitary $U_{W}$]\label{def:s-unitary}
For any $\sec\in\NN$ and
$W\subseteq\bit^\secp\setminus \{0^\secp\}$, we define the states
\begin{align*}
    |W\rangle=\frac{1}{\sqrt{|W|}}\sum_{x\in W}|x\rangle &&\text{and}&& \ket{W-} = \frac{1}{\sqrt{2}} \left( \ket{W} - \ket{0^\secp} \right).
\end{align*}
We define the $(\secp+1)$-qubit unitary $U_{W}$ as follows $$U_{W}=I-2\ket{1}\proj{W-}\bra{1}.$$ Namely, $U_{W}$ is the controlled reflection unitary about $\ket{W-}$, i.e. $R^c_{W-}$.

Note that, $R_{W-}$ maps $\ket{W}$ to $\ket{0^\sec}$ and $\ket{0^\sec}$ to $\ket{W}$, and acts as identity on the subspace orthogonal to $\Span(\ket{0^\sec},\ket{W})$.



\end{definition}

We now define the distribution $\oracle$ over oracles $O$ where, roughly speaking, a random oracle $O\gets\oracle$ applies the unitary $U_S$ over its $(\sec+1)$-qubit input, for a uniformly random subset $S$ of size $2^{{\secp}/{2}}$.

\begin{definition}[The Oracle $\mathcal{O}$]\label{def:S-oracle}
We denote by $\oracle$ the distribution over quantum oracles where
\begin{itemize}
    \item \textbf{Randomness:} A random $O_S\gets\oracle$ is defined by an ensemble $S=\{S_\sec\}_{\sec\in\NN}$ where, for any $\sec\in\NN$, $S_\sec$ is a uniformly random subset in $\bin^\sec\setminus \{0^\sec\}$ of size $|S_\sec|=2^{\frac{\secp}{2}}$.
    \item \textbf{Query:} For any $S=\{S_\sec\}_{\secp\in\mathbb{N}}$, $\sec\in\NN$ and $(\sec+1)$-qubit input state $\rho$, we define $O_S(\rho)=U_{S_\secp}\rho U_{S_\secp}^\dagger$.
    
\end{itemize}

\end{definition}

A key ingredient in our proofs is the application of the reflection emulation technique from~\cite{C:JiLiuSon18} to our oracle $\oracle$. Since any $O\in\oracle$ essentially computes a reflection unitary about the state $\ket{S_\sec-}$ on any input of length $\sec+1$ (where $S_\sec$ is the random subset underling $O$), reflection emulation (namely \cref{thrm:reflect-jls}) tells us that we are able to simulate the oracle given copies of the corresponding state $\ket{S_\sec -}$. The following proposition follows from \cref{thrm:reflect-jls,def:s-unitary}.

\begin{proposition}[Emulating $\oracle$]\label{prop:reflect-emulate-efid}
Let $\epsilon>0$, $m\in \mathbb{N}$ and $\secp_1,\ldots,\secp_m\in \mathbb{N}$. Let $q\in \mathbb{N}$ and define $t={2q^2}/{\epsilon^2} -1$. For every $i\in[m]$, let $S_{i}$ be a subset of strings in 
$\bin^{\secp_i}\setminus\{0^{\secp_i}\}$, and $U_{S_{i}},\ldots, U_{S_{m}}$ be as defined in \Cref{def:s-unitary} and $Q$ be any other quantum oracle.
    Let $\cA^{Q,U_{S_{i}},\ldots, U_{S_{m}}}$ be an oracle-aided quantum algorithm that makes at most $q$ queries in total to the oracles $U_{S_{i}},\ldots, U_{S_{m}}$. 
    Then, there exists a quantum algorithm $\cB$ such that, for any input state $\rho$, 
    \begin{align*}
    \TD\left[\left(\cA^{Q,U_{S_{1}},\ldots, U_{S_{m}}}(\rho)\bigotimes_{i=1}^m\ket{S_{i}-}^{\otimes t}\right),\ \cB^Q\left(\rho,\bigotimes_{i=1}^m\ket{S_{i}-}^{\otimes t}\right)\right]
    \leq \epsilon.
    \end{align*}
    Moreover, the running time of $B$ is polynomial in that of $A$ and $\ell$.
\end{proposition}

Lastly, we recall the definition of the QPSPACE oracle from~\cite{CCS24}.

\begin{definition}[QPSPACE Oracle~\cite{CCS24}]\label{def:qpspace-machine-oracle}
    We define the \emph{QPSPACE machine oracle}, which we denote by $\qpspace$, as follows. The oracle takes as input an $\ell$-qubit state $\rho$, a description of a classical Turing machine $M$ and an integer $t\in\NN$. The oracle runs $M$ for $t$ steps to obtain the description of a quantum circuit $C$ that operates on exactly $\ell$ qubits. If $M$ does not terminate after $t$ steps, or if the output circuit $C$ does not operate on $\ell$ qubits, the oracle returns $\bot$. Otherwise, the oracle applies $C$ on $\rho$ and returns the output without measurement.
\end{definition}

We stress that we restrict ourselves to a ``unitary world''. That is, our oracles define unitary operations and, additionally, we always provide any quantum algorithm in our oracle-relative world with access to a unitary \emph{and its inverse}. In our case, however, this is w.l.o.g. since the oracle $O$ is equivalent to its inverse and the inverse of $\qpspace$ can be simulated by a single query to $\qpspace$.

\begin{restatable}{proposition}{qpspaceinverse}\label{prop:qpspace-inverse}
    There exists a $\qpspace$-aided QPT algorithm which, on any $\qpspace$-input $(\rho,M,t)$ where $\rho$ is an $\ell$-qubit state, $M$ is the description of a Turing machine and $t\in\NN$ (see \cref{def:qpspace-machine-oracle}), outputs $\qpspace^{-1}(\rho,M,t)=C^{-1}(\rho)$, where $C$ is the quantum circuit whose description is output by the machine $M$ after running it for $t$ steps.
\end{restatable}

\submitorfull{A proof for \Cref{prop:qpspace-inverse} is given in \Cref{sec:oracles-proof}}{qpspace-inverse}.

\section{Existence of \qefid Pairs}\label{sec:efid}

As a first step towards proving \cref{thrm:efi-exist}, we demonstrate that an exponentially-hard \qefid pair (as per \cref{def:efi}) exists relative to the oracles we defined in \cref{sec:oracles}, namely a random $O\gets \oracle$ and $\qpspace$.

\begin{theorem}[\qefid pairs exist under $(\oracle,\qpspace)$]\label{thrm:efi-exist}
   Let the oracles $\cO,\qpspace$ be as defined in \Cref{def:S-oracle,def:qpspace-machine-oracle}, respectively.
   Then, with probability $1$ over the choice of $O\gets\oracle$, exponentially-hard \qefid pairs, more precisely $(2^{\sec/50},2^{-\sec/50})$-\qefid pairs, exist relative to $(O,\qpspace)$.
\end{theorem}

We start by describing our oracle-relative candidate \qefid pair.

\begin{construction}[\qefid Pair under $\oracle$]\label{con:efi}
    We construct an $\oracle$-aided \qefid pair $(D_0^\cO,D_1^\cO)$ which, given $O\in\oracle$, acts as follows:
    \begin{itemize}
        \item []$D_0^O(1^\sec)$:
        Query $O$ with $|1\rangle|0^{\secp}\rangle$ to get $\ket{1}|S_{\secp}\rangle$. 
        Measure the second register in the computational basis
        to obtain a $\secp$-bit string $s\in S_\secp$. Output it.
        \item[]$D_1^O(1^\sec)$:
        Output a uniformly random $u\gets\bit^\secp$.
    \end{itemize}
\end{construction}

It is immediate, by construction, that the \qefid pair from \cref{con:efi} is efficiently samplable for any $O\in\oracle$. It remains, then, to show that it satisfies statistical farness and computational indistinguishability relative to $(O,\qpspace)$ (see \cref{def:efi}). Let us start with the former.

\begin{restatable}[Statistical Farness]{proposition}{statisticalfarness}
    \label{prop:statistical-farness}
    For any fixed $O\in\cO$, the distributions $D_0^O(1^\secp)$ and $D_1^O(1^\secp)$ defined in \Cref{con:efi}, are statistically far.
\end{restatable}

\submitorfull{A proof of \Cref{prop:statistical-farness} is given in \Cref{sec:efid-proofs}.}{statistical-farness} Next, we argue computational indistinguishability of our construction through the following lemma.

\begin{lemma}[Computational Indistinguishability]\label{lem:computational-indistinguishability-EFID-pair}
    With probability 1 over the choice of $O\gets\oracle$, for any non-uniform $(O,\qpspace)$-aided adversary $\adv$, that on any security parameter $\sec$ makes at most $O(2^{\sec/50})$ calls to the oracle $O$ (and is otherwise unbounded), it holds that
    \[
        \left|\Pr_{z\gets D_0(1^\secp)}[ \adv^{O,\qpspace}(1^\secp,z)=1]-\Pr_{z\gets D_1(1^\secp)}[ \adv^{O,\qpspace}(1^\secp,z)=1]\right|= O(2^{-\sec/50}),
    \]
    for any $\sec\in\NN$.
\end{lemma}

Note that we show computational indistinguishability even against adversaries that are unbounded in their runtime or queries to the $\qpspace$ oracle. While this is stronger than what we need for the oracle separation between OWSGs and \qefid pairs, 
it will become useful for our separation concerning private-key quantum money schemes (\Cref{thm:efid-vs-money}).

The proof of \cref{thrm:efi-exist} follows immediately by combining \Cref{prop:statistical-farness,lem:computational-indistinguishability-EFID-pair}. In what comes next, we prove \cref{lem:computational-indistinguishability-EFID-pair}.

By reflection emulation (see \cref{prop:reflect-emulate-efid}), we may emulate any $(O,\qpspace)$-aided adversary by an adversary that does not have access to the oracle yet requires copies of the state $\ket{S_\sec-}$, where $S_{\sec}$ is the random subset underlying $O$. This reduces our task to showing computational indistinguishability against such adversaries. By replacing access to the oracle $O$ with copies of the state $\ket{S_\sec-}$, our goal becomes showing that it is hard to break our \qefid given such copies. At the core of such an argument, then, is the following statistical lemma.

\begin{restatable}{lemma}{statisticallemma}
    \label{lem:EFID-trace-distance}
    Let $\sec\in\NN$ and let $S,W\subset \bin^\sec\setminus\{0^\sec\}$ be two uniformly random subsets of size $2^{\frac{\secp}{2}}$, $s\gets S$, and $u\gets\bin^\sec$. Then, it holds that
    \begin{equation*}
    \TD\left[\left\{\ket{S-}^{\otimes t}\otimes\ket{s}\right\}_{S,s},\left\{\ket{W-}^{\otimes t}\otimes \ket{u}\right\}_{W,u}\right]<2^{-\frac{\secp}{2}+1} +\sqrt{t\cdot 2^{-\frac{\secp}{2}}}.
    \end{equation*}
\end{restatable}

\submitorfull{We prove \Cref{lem:EFID-trace-distance} in \Cref{sec:efid-proofs}.}{statistical-lemma} Putting \cref{prop:reflect-emulate-efid,lem:EFID-trace-distance} together, we obtain the following proposition\ifnum\submission=1
, which we also prove in \Cref{sec:efid-proofs}.
\else
.
\fi

\begin{restatable}{proposition}{efidemulator}
    \label{prop:emulator-step-EFID}
    For any $(\cO,\qpspace)$-aided quantum algorithm $\adv$ that makes at most $O(2^{\sec/50})$ queries to $\cO$, it holds that
    \[
        \EE_{O\gets \cO}\left[\Pr_{s\gets D^O_0(1^\secp)}[\adv^{O,\qpspace}(1^\secp,s)=1]-\Pr_{s\gets D^O_1(1^\secp)}[\adv^{O, \qpspace}(1^\secp,s)=1]\right]\leq 2^{-{2\secp}/{25}}.
    \]
\end{restatable}

\submitorfull{}{efid-emulator}

While it might seem like \cref{prop:emulator-step-EFID} implies computational indistinguishability of our \qefid construction for a random oracle $O\gets\oracle$, this is not immediately the case since we have bounded the positive gap between the probability that the distinguisher's output is 1 at input from $D_0$ compared to an input from $D_1$. Indistinguishability, however, requires bounding the absolute value between the two probabilities.

It is a well-known fact that such a ``positive-gap'' indistinguishability implies standard absolute-value indistinguishability~\cite{Yao82,BG11} (with some loss in advantage). This has been further extended to the case of a quantum oracle-aided distinguisher (where the gap is guaranteed in expectation over a random oracle, just as we require).

\begin{restatable}[Absolute-Gap to Positive-Gap Distinguisher]{lemma}{abstopos}
\label{lem:abs-to-pos}
    Let $\oracle$ be a distribution over oracles. Let $D^\oracle_0$ and $D^\oracle_1$ 
    be two $\oracle$-aided distributions over $\bin^\sec$ with corresponding sampling algorithms. Let $\adv$ be an $\oracle$-aided quantum algorithm such that
    \[
       \Expct_{O\gets\oracle}\left[\ \left| \Pr_{x\gets D_0}[\adv(x)=1] - \Pr_{x\gets D_1}[\adv(x)=1]\right|\ \right] = \delta.
    \]
    
    Then, there exists an $\oracle$-aided quantum algorithm $\badv$ such that
    \[
        \Expct_{O\gets\oracle}\left[\Pr_{x\gets D_0}[\badv(x)=1] - \Pr_{x\gets D_1}[\badv(x)=1]\right] \geq \delta^2.
    \]
    The runtime and query complexity of $\badv$ is twice that of $\adv$ in addition to that of the sampling algorithms of $D_0$ and $D_1$.
\end{restatable}

For completeness, we attach a proof to the lemma in \cref{app:abs-to-pos}. As a corollary, we obtain the following.
\begin{corollary}\label{cor:absolute-value-distinguisher}
    For any $(\cO,\qpspace)$-aided quantum algorithm $\adv$ that makes at most $O(2^{\sec/50})$ queries to $\cO$, it holds that
    \[
        \EE_{O\gets \cO}\left[\left|\Pr_{s\gets D^O_0(1^\secp)}[\adv^{O,\qpspace}(1^\secp,s)=1]-\Pr_{s\gets D^O_1(1^\secp)}[\adv^{O, \qpspace}(1^\secp,s)=1]\right|\right]\leq 2^{-{\secp}/{25}}.
    \]
\end{corollary}

With \cref{cor:absolute-value-distinguisher} in hand, we proceed to complete the proof of \cref{lem:computational-indistinguishability-EFID-pair}.

Fix an $(\oracle,\qpspace)$-aided quantum adversary $\adv$ that makes at most $O(2^{{\secp}/{50}})$ queries to $\oracle$. By \cref{cor:absolute-value-distinguisher} and Markov inequality, we conclude that
\begin{align*}
\Pr_{O\gets\oracle}\left[\left|\Pr_{s\gets D^O_0(1^\secp)}[\adv^{\cO,\qpspace}(1^\secp,s)=1]
-\Pr_{s\gets D^O_1(1^\secp)}[\adv^{O, \qpspace}(1^\secp,s)=1]\right|\geq 2^{-{\secp}/{50}} \right]\leq2^{-{\secp}/{50}}.
\end{align*}

Since $\sum_\secp2^{-{\secp}/{50}}$ converges, by Borel-Cantelli Lemma we have that, with probability $1$ over the choice of $O\gets\oracle$, it holds that
\begin{align}\label{eq:A-bound}
\left|\Pr_{s\gets D^\cO_0(1^\secp)}[\adv^{\cO,\qpspace}(1^\secp,s)=1]-\Pr_{s\gets D^\cO_1(1^\secp)}[\adv^{\cO, \qpspace}(1^\secp,s)=1]\right|\leq 2^{-{\secp}/{50}},
\end{align}
except for finitely many $\secp\in \mathbb{N}$.
Since there are only countably many such quantum algorithms $\adv$ that make at most $O(2^{{\secp}/{50}})$ queries to $\cO$, we conclude that with probability $1$ over the oracles 
$(\cO,\qpspace)$, it holds that:
for every quantum algorithm $\adv$ making at most $O(2^{{\secp}/{50}})$ queries to $\cO$, \Cref{eq:A-bound} holds, which completes the proof of \cref{lem:computational-indistinguishability-EFID-pair}.

\section{Impossibility of OWSGs}

In this section, we show that, for any fixed $O\in\cO$,  
OWSGs (as defined in \cref{def:owsg}) do not exist in the presence of the oracles $(O,\qpspace)$. Together with \cref{thrm:efi-exist}, this completes the proof of our first main result from \cref{thm:efid-vs-owsg}, namely, the oracle separation of OWSGs from \qefid pairs.

\begin{theorem}\label{thm:break-OWSG}
For any $O\in\cO$, OWSGs do not exist relative to the oracles $(O, {\bf QPSPACE})$.
\end{theorem}

\subsection{Gentle Search for QPSPACE-aided POVMs}\label{sec:gentle-search}


A key ingredient for breaking any oracle-relative OWSG candidate is the \emph{gentle search} procedure~\cite{Shadow2}. In our context, gentle search allows, given a quantum state and a collection of verification keys, to identify a key under which the state is accepted (by an apriori-fixed verification algorithm).

While various algorithms for standard gentle search exist in the literature~\cite{Shadow2,BW24}, we want to additionally allow the verification algorithm to have access to the $\qpspace$ oracle. To this end, via an observation made in~\cite{CCS24}, we adapt the algorithm from~\cite{Shadow2}, in particular its \emph{OR-tester component}~\cite{HLM17}, to work also when the verification algorithm has polynomially-bounded access to $\qpspace$. Overall, we obtain the following general gentle search algorithm for $\qpspace$-aided POVMs.

\begin{restatable}[Gentle Search via QPSPACE Machine]{lemma}{gentlesearch}
    Let $K$ be a finite set of strings and let $\{\Pi^{\qpspace}_k\}_{k\in K}$ be a family of $\qpspace$-aided binary-valued POVMs, indexed by elements in $K$, each of which makes polynomially many queries to its oracle. 
    Suppose $\ket{\psi}$ is a state such that there exist a real $c>0$ and a key $k\in K$ for which $\Tr(\Pi^\qpspace_k\ket{\psi}\bra{\psi})\geq c$.
    Let $\epsilon,\delta >0$. Then, there exists a polynomial-time $\qpspace$-{aided} quantum algorithm, which we denote by $\tomography$, that takes as input $\ket{\psi}^{\otimes t}$, for $t\in  O\left((\log^4 |K|\log\log|K|+\log(1/\delta))/{\epsilon^2}\right)$, makes $\log(|K|)$ queries to its oracle and outputs a key $k'\in K$ such that, with probability at least $1-\delta$,
    \ifnum\submission=0
    
    \begin{align*}
    \Tr(\Pi^\qpspace_{k'}\ket{\psi}\bra{\psi})\geq c-\epsilon.
    \end{align*}

    \else
        $\Tr(\Pi^\qpspace_{k'}\ket{\psi}\bra{\psi})\geq c-\epsilon$.
    \fi
    \label{lem:gentle-search}
\end{restatable}

\submitorfull{We provide a detailed proof to \Cref{lem:gentle-search} in \Cref{sec:gentle-search-proof}}{gentle-search}.

\subsection{Proof of \texorpdfstring{\Cref{thm:break-OWSG}}{Theorem 5.1}}

Fix $O\in\oracle$ and let $(\G,\V)$ be a OWSG candidate relative to $(O,\qpspace)$ (see \cref{def:owsg}). We show that the fact that $\V$ is a QPT (in particular, makes a polynomial number of queries to its oracles) necessarily implies an attack against the one-wayness of the candidate.

Our attack is based on the gentle search algorithm as implied by \cref{lem:gentle-search}; Given (many copies of) a state $\phi_{k^*}$ corresponding
to a key $k^*$ and the
collection of POVM elements $\{\Pi_k\}_k$, where $\Pi_k$ corresponds to the event of $\V$ accepting with key $k$,
it is possible with access to the $\qpspace$ oracle to identify a key $k$ such that $(\phi_{k^*},k)$ is accepted by $\V$ with good probability. 

The main obstacle in this outline is that the algorithm $\V$ is oracle-aided. While \cref{lem:gentle-search} allows the POVMs to be $\qpspace$-aided, the oracle $O$ remains an issue. Luckily, due to \cref{prop:reflect-emulate-efid}, we know that polynomially-many calls to $O$ with queries of any length $\sec\in\NN$ can be emulated by polynomially-many copies of the state $\ket{S_\sec-}$, where $S=\{S_\sec\}$ are the random subsets underling the oracle $O$. In particular, we can simulate $\V$ using such copies for any $\sec$ w.r.t. which a query is made by $\V$, without accessing $O$.\footnote{Note that we use $\kappa$ to denote the security parameter of the OWSG we aim to break. $\sec$ is used to denote the security parameter w.r.t. which our attack invokes its queries to the oracle $O$. While polynomially-related, $\kappa$ and $\sec$ are not necessarily equal.}

Let $q:=q(\kappa)$ be the polynomial bound on the query complexity of $\V$ and let $\overline{\sec}:=\overline{\sec}(\kappa)$ be the polynomial bound on the length of these queries. That is, for any $\kappa\in\NN$ and any $k\in\bin^\kappa$, $\V(k,\cdot)$ always makes at most $q(\kappa)$ queries to $O$, where each is of length at most $\overline{\sec}(\kappa)$. Then, by \Cref{prop:reflect-emulate-efid}, there exists a polynomial-query $\qpspace$-aided algorithm $\widetilde{\V}$ such that, for any $k\in\bin^\kappa$ and state $\phi$,
\begin{equation}\label{eq:v-to-tilde}
    \left| \Pr[\V^{O,\qpspace}(k,\phi)=1] -\Pr[\widetilde{\V}^{\qpspace}(k,\phi\otimes\bigotimes_{\sec=1}^{\overline{\sec}}\ket{S_\sec -}^{\otimes 2q^2\kappa^2})=1] \right| \leq 1 /\kappa.
\end{equation}

Given the above, to apply gentle search w.r.t. the POVMs defined by $\widetilde{\V}$ we no longer need access to the oracle $O$. Instead, we must generate many copies of $\ket{S_\sec -}$. To that end, we use the fact that $O$ contains the controlled reflection about $\ket{S_\sec-}$ (see \cref{def:S-oracle}); due to \cref{prop:project-via-reflect}, there exists a single-query algorithm which, on input $\ket{0^\sec}$ and access to $O$, outputs $\ket{S_\sec-}$ with probability $|\braket{0}{S_\sec-}|^2=1/2$ and otherwise declares failure. By repeating such a procedure for $\kappa$ iterations or until success, we get the following corollary.

\begin{corollary}\label{cor:phase-estimation}
    For any $\kappa\in\NN$, there exists an $\oracle$-aided quantum algorithm $E$ that makes $\kappa$ queries on the input $1^\lambda$, such that, for any $O\in\oracle$ and $\sec\in\NN$,
    \[
        \Pr[E^O(1^\sec)=\ket{S_\sec -}]\geq 1 - 2^{-\kappa},
    \]
    where $S_\sec\subset \bin^\sec \bin^\sec\setminus\{0^\sec\}$
    is the subset underlying the oracle $O$ (see \cref{def:S-oracle}).
\end{corollary}

We are now prepared to describe the attack on the OWSG candidate that requires $t(\kappa)\in O(\kappa^2\log \kappa)$ copies of the challenge state.
\begin{itemize}
    \item[] $\adv^{O,\qpspace}(1^\kappa,\phi^{\otimes t})$: 
\begin{enumerate}
    \item For $\sec=1,\dots,\overline{\sec}$, run $E^O(1^\sec)$ (from \cref{cor:phase-estimation}) $2q^2\kappa^2t$ times to obtain $\ket{S_\sec-}^{\otimes 2q^2 \kappa^2 t}$.
    \item Run $\qpspace$-aided gentle search, namely the algorithm $\tomography$ from \cref{lem:gentle-search}, with
    \if\submission=0
    \begin{itemize}
        \item parameters $\epsilon=\delta=1/2$. 
        \item POVMs $\{\tilde{\V}^\qpspace(k,\cdot)\}_{k\in\bin^\kappa}$, and
        \item input $\phi\otimes\left(\bigotimes_{\sec\in[\overline{\sec}]} \ket{S_{\sec}-}^{\otimes 2q^2\kappa^2}\right)$,
    \end{itemize} and output its outcome.
    \else
    parameters $\epsilon=\delta=1/2$, POVM elements $\{\Pi_k\}_{k\in\bin^\kappa}$ (where $\Pi_k$ corresponds the acceptance of
    $\tilde{\V}^\qpspace(k,\cdot)$), and input 
    $\phi\otimes\left(\bigotimes_{\sec\in[\overline{\sec}]} \ket{S_{\sec}-}^{\otimes 2q^2\kappa^2}\right)$ (of which we have $t$ copies -- this suffices due to \Cref{lem:gentle-search}), and output its outcome.
    \fi
\end{enumerate}
\end{itemize}

By \cref{cor:phase-estimation}, the first step succeeds with probability at least $1-2q^2\kappa^2\overline{\sec}/2^\kappa$. Now, assuming $\phi$ was sampled by the OWSG under key $k^*\in\bin^\kappa$, it holds by correctness that 
$$\Pr[{\V}^{O,\qpspace}(k^*,\phi)=1]\geq 1-\negl(\kappa),$$ 
for some negligible function $\negl$, and hence, by \cref{eq:v-to-tilde}, 
$$\Pr[\tilde{\V}^{\qpspace}(k^*,\phi\otimes\bigotimes_{\sec=1}^{\overline{\sec}}\ket{S_\sec-}^{\otimes 2q^2\kappa^2})=1]\geq 1-1/\kappa-\negl(\kappa).$$ 
Given the above guarantee regarding $\phi$, \cref{lem:gentle-search} implies that the second step finds a $k'\in\bin^\kappa$ such that 
$$\Pr[\tilde{\V}^{\qpspace}(k',\phi\otimes\bigotimes_{\sec=1}^{\overline{\sec}}\ket{S_\sec-}^{\otimes 2q^2\kappa^2})=1]\geq 1/2-1/\kappa-\negl(\kappa) > 1/3$$ 
and, consequently, 
$\Pr[{\V}^{O,\qpspace}(k',\phi)=1]\geq 1/3-1/\kappa$. This concludes the attack on the OWSG candidate successful and completes the proof of \cref{thm:break-OWSG}.

\ifnum\submission=0

\begin{remark}
    While we describe an attack against the OWSG that succeeds with probability $1/3-1/\kappa$ for security parameter $\kappa$, this is merely for simplicity of exposition. A more careful tuning of the parameters results in an attack that succeeds with probability $1-\Theta(1/\kappa)$.
\end{remark}

\fi

\ifnum\submission=0

\section{Further Implications of \texorpdfstring{\Cref{thm:efid-vs-owsg}}{Theorem 1.1}}\label{sec:further-implications}

 Firstly, observe that EV-OWPuzzs (see \Cref{def:EV-OWPuzz}) are essentially a classical state version of OWSGs, up to a change of syntax for the sampling the state/puzzle. It was shown in~\cite{C:ChuGolGra24} that these syntaxes are equivalent, and hence EV-OWPuzzs imply OWSGs. 
\begin{theorem}[{Implicit in~\cite[Theorem 32, section 10]{C:ChuGolGra24}}]\label{thm:EV-OWPuzz-implies-OWSG}
    EV-OWPuzzs (see \Cref{def:EV-OWPuzz}) imply OWSGs (see \Cref{def:owsg}) in a black-box manner.
\end{theorem}

Combining \Cref{thm:EV-OWPuzz-implies-OWSG} with \Cref{thm:break-OWSG,thrm:efi-exist}, we get the following corollary.
\begin{corollary}[Restatement of \Cref{cor:EFID-EV-OWPuzz}]
     There exists a unitary oracle relative to which \qefid pairs (see \Cref{def:efi}) exist, but EV-OWPuzzs do not (see \Cref{def:EV-OWPuzz}).
\end{corollary}
It was also shown in~\cite{C:ChuGolGra24} that \qefid pairs imply OWPuzzs.
\begin{theorem}[{Implicit in~\cite[Lemma 8]{C:ChuGolGra24}}]\label{thm:qefid-implies-OWPuzz}
   \qefid pairs (see \Cref{def:EV-OWPuzz}) imply OWPuzzs (see \Cref{def:owsg}) in a black-box manner.
\end{theorem}
Combining \Cref{thm:qefid-implies-OWPuzz} with \Cref{thm:break-OWSG,thrm:efi-exist}, we get the following corollary.
\begin{corollary}[Restatement of \Cref{cor:owpuzz-owsg}]
     There exists a unitary oracle relative to which OWPuzzs (see \Cref{def:EV-OWPuzz}) exist, but OWSGs 
     (see \Cref{def:EV-OWPuzz})
     do not.
\end{corollary}

Next, we move deeper into Microcrypt, and start with the following simple observation.
\begin{theorem}\label{thm:Strong-uncl-states-imply-EVOWSG}
    Strongly unclonable state generators $(\struncl)$ imply (efficiently verifiable) OWSGs.
\end{theorem}

\begin{proof}
    Let $\G$ be the QPT generation algorithm of an UCSG with key space $\keyspace=\{\keyspace_\kappa\}_\kappa$ (see \cref{def:unclonable}). Let $\G^\kappa$ denote the algorithm which, on input a key $k\in\keyspace_\kappa$, for any $\kappa\in\NN$, invokes $\G(k)$ $\kappa$ times and outputs the tensor of all $\kappa$ output states. Let $\V$ be the algorithm that takes as input a key $k\in\keyspace_\kappa$ and a state $\psi$, invokes $\G^\kappa(k)$ to obtain $\ket{\phi_k}^{\otimes \kappa}$, and applies a swap test over each of the copies of $\ket{\phi_k}$ and the corresponding register in $\psi$. $\V$ outputs 1 if all tests succeed.

    We argue that the pair $(\G^\kappa,\V)$ constitutes an efficiently verifiable OWSG. First, correctness holds since, for any $k\in\keyspace_\kappa$, $\Pr[\V(k,\G^\kappa(k))=1] = (\frac{1}{2}(1+|\braket{\phi_k}{\phi_k}|^2))^\kappa = 1$, where $\ket{\phi_k}=\G(k)$. For security, let $\adv$ be an adversary that breaks the one-wayness of $(\G^\kappa,\V)$. That is, there exists a polynomial $t$ and a non-negligible function $\eta$  such that, for a random key $k\gets\keyspace_\kappa$ and the corresponding state $\ket{\phi_k}= \G(k)$, it holds that $\adv(1^\kappa,\ket{\phi_k}^{\otimes \kappa\cdot t})$ outputs a $k'$ such that $\V(k',\ket{\phi_k}^{\otimes \kappa})=1$ with probability $\eta:=\eta(\kappa)$. The latter necessarily implies that, for infinitely many $\kappa$, 
    $|\langle\phi_{k'}|\phi_k\rangle|^2>0.8$ with probability at least $\eta/2$ since, otherwise, 
    $$\Pr[\V(k',\ket{\phi_k}^{\otimes \kappa})=1] <\eta/2+ (1-\eta/2)\left(\frac{1}{2}(1+|\braket{\phi_{k'}}{\phi_k}|^2)\right)^\kappa < \eta/2 + 0.9^{\kappa}(1-\eta/2)<\eta$$
    for all but finitely many $\kappa$.

    Given the above, we derive an attack against the presumed unclonable state family. The attack, which we recall has access to the controlled reflection oracle $R^c_{\phi_k}$, takes as input security parameter $1^\kappa$ and $\kappa\cdot t$ copies of the state $\ket{\phi_k}$ and acts as follows:
    \begin{enumerate}
        \item Run $\adv(1^\kappa,\ket{\phi_k}^{\otimes \kappa\cdot t})$ to obtain a key $k'$.
        \item Run the following for at most $\kappa^2\cdot t$ iterations or until $\kappa\cdot t + 1$ iterations succeed:
        \begin{enumerate}[label*=\arabic*.]
            \item Generate the state $\ket{\phi_{k'}}\gets\G(k')$.
            \item Invoke the algorithm from \cref{prop:project-via-reflect} with input $\ket{\phi_{k'}}$ using the oracle $R^c_{\phi_k}$. If the algorithm fails, abort. Otherwise, the output state is $\ket{\phi_k}$.
        \end{enumerate}
        \item Output the $\kappa\cdot t+1$ copies of $\ket{\phi_k}$ generated by the loop in the previous step.
    \end{enumerate}

    We now analyze the success probability of our proposed attack. First, recall that for infinitely many $\kappa$, it holds $|\langle\phi_{k'}|\phi_k\rangle|^2>0.8$ with probability at least $\eta/2$. Assuming this event occurs, due to \cref{prop:project-via-reflect}, every iteration of step 2.2. in the attack is successful with probability at least $0.8$. Hence, by a standard Chernoff bound, the probability that $\kappa\cdot t+1$ attempts among the $\kappa^2\cdot t$ iterations succeed is at least $1-2^{-\kappa}$. This concludes the proof.

\end{proof}
Note that in the proof of \Cref{thm:Strong-uncl-states-imply-EVOWSG}, it was crucial that the adversary $\adv$ has access to the controlled-reflection about the state, i.e., the state family is a \emph{strongly} unclonable family. Next, we use our emulation result for arbitrary reflections to show that this additional access to the controlled-reflection does not add any strength.
\begin{theorem}\label{thm:uncl-state-imply-strong-uncl-states}
Any UCSG is also strongly unclonable.
\end{theorem}
\begin{proof}
    Let $\Phi=\{\ket{\phi_k}\}_{k\in\keyspace}$ be an unclonable state generator. Assume towards contadiction that it is not strongly unclonable. Then, there exist an oracle-aided QPT algorithm $\adv$ and polynomial $m:=m(\kappa)$, 
     such that
    \begin{equation}\label{eq:uncl-advantage}
        \Expct_{k\gets\keyspace_\kappa}\left[\Tr\left( \ketbra{\phi_k}{\phi_k}^{\otimes m+1}\cA^{R^c_{\phi_k}}(1^\kappa,|\phi_k\rangle^{\otimes m})\right)\right]>\eta(\kappa)
    \end{equation}
    for infinitely many $\kappa$, where $\eta:=\eta(\kappa)$ is a non-negligible function and $R^c_{\phi_k}$ is the controlled reflection oracle about $\phi_k$. Let $q:\NN\to\NN$ be a polynomial bound on the query complexity of $\adv$, i.e. $\adv$ makes at most $q:=q(\kappa)$ queries to its oracle on input $1^\kappa$. Then, due to \cref{thrm:reflect-jls}, we may simulate $\adv$ by a QPT algorithm $\badv$ that does not have access to $R^c_{\phi_k}$ but requires additional copies of $\ket{\phi_k}$.\footnote{In fact,  by \cref{thrm:reflect-jls}, $\badv$ requires additional copies of $\ket{\phi'_k}=\ket{1}\otimes\ket{\phi_k}$, since $R^c_{\phi_k}=R_{\phi'_k}$. However, $\badv$ may easily transform $\ket{\phi_k}$ into $\ket{\phi'_k}$.} More concretely, $\badv$ satisfies the following
    \begin{equation}\label{eq:uncl-advantage-emulator}
        \TD\left(\adv^{\reflect^c_{\phi_k}}(1^\kappa,\ket{\phi_k}^{\otimes m})\otimes\ket{\phi_k}^{\otimes 8q^2/\eta^2}, 
        \badv(1^\kappa,\ket{\phi_k}^{\otimes m}\otimes \ket{\phi_k}^{\otimes 8q^2/\eta^2})\right)\leq \eta/2.
    \end{equation}

    By combining \cref{eq:uncl-advantage,eq:uncl-advantage-emulator}, $\badv$ constitutes a successful attack against the (standard) unclonability of $\Phi$: On input $m+8q^2/\eta^2$ copies of $\ket{\phi_k}$, it produces $m+1+8q^2/\eta^2$ with advantage at least $\eta/2$.
\end{proof}

\begin{theorem}[{\cite[Theorem 2]{C:JiLiuSon18}}]\label{thm:PRS-implies-uncl-states}
    PRSGs imply UCSGs.
\end{theorem}
\begin{remark}
    In~\cite[Theorem 5]{C:JiLiuSon18}, it was shown that PRSGs imply strong UCSGs. 
    However, they did not show that UCSGs imply strong UCSGs, which we observe in \Cref{thm:uncl-state-imply-strong-uncl-states}. 
\end{remark}


\begin{corollary}
     For any fixed $O\in\oracle$, strong UCSGs, UCSGs and PRSGs do not exist relative to $O, \qpspace$. 
     Hence, we conclude that strong UCSGs, UCSGs and PRSGs are separated from \qefid pairs (see \Cref{def:efi}).
\end{corollary}
The proof is immediate by combining \Cref{thm:PRS-implies-uncl-states,thm:uncl-state-imply-strong-uncl-states,thm:Strong-uncl-states-imply-EVOWSG} with \Cref{thm:break-OWSG}.

\begin{remark}\label{rem:uncl-non-telegraphable-one-wayness}
    An incomparable but related result is~\cite{NZ24} where the authors separated UCSGs (or unclonable states as referred to in~\cite{NZ24}) and non-telegraphable states. Since non-telegraphy seems to imply one-wayness, which implies EFI pairs~\cite{STOC:KhuTom24}, intuitively, it seems that \cite{NZ24,STOC:KhuTom24} together imply a separation between UCSGs and EFI pairs. However, this is not true because: 1) the definition of non-telegraphable and unclonable states considered in~\cite{NZ24}, only allows the adversary to get a single copy of the state in the respective security games, due to which the resulting OWSGs from non-telegraphable states are also single-copy secure, but all existing constructions~\cite{STOC:KhuTom24,BJ24} of EFI pairs from OWSGs require multi-copy security, 2) we consider a multi-copy definition of UCSGs where the cloner in the cloning game can get any polynomial number of copies of the state, hence ruling out single-copy secure UCSGs as shown in~\cite{NZ24} does not rule out all UCSGs as per our definition.
    
\end{remark}

\fi

\section{Separating Private-key Quantum Money from {\qefid}}\label{sec:separating-qmoney}


In this last section, we prove our second main separation result, namely that between private-key quantum money schemes and \qefid pairs. Here, in contrast to our oracle separation of OWSGs from \qefid pairs that is laid down in the previous sections, we are able to establish a weaker notion of separation, specifically a \emph{fully black-box separation}.

Fully black-box separation means the impossibility of fully black-box constructions. These include any realization of a cryptographic primitive based on another (referred to as the base primitive), where the construction uses the algorithms underlying the base primitive in a black-box way (i.e. independently of their implementation) and, additionally, the security reduction breaks any implementation of the base primitive given a black-box access to an adversary that breaks the construction based on that implementation. For a formal definition, we refer the reader to \Cref{sec:fbb-def}.

To rule out a fully black-box construction of private-key quantum money schemes from \qefid pairs it is sufficient, then, to devise an oracle world where \begin{enumerate*}[label=(\roman*)]
    \item \label{item:efid-exist} there exists a \qefid construction that is secure against polynomial-query adversaries, even if unbounded in runtime, and, on the other hand,
    \item any private-key quantum money scheme may be broken by a polynomial-query attack that is possibly unbounded otherwise.
\end{enumerate*} This is a relaxation compared to an oracle separation where we consider polynomially bounded adversaries, both in queries and runtime (where the ``unboundedness'' may be ``pushed'' entirely to the oracles). For an explanation on why we are able to achieve only a fully black-box separation in the private quantum money case, and the challenges in obtaining an oracle separation, we refer the reader to the technical overview in \Cref{sec:techniques}.

The oracle world through which we separate between private-key quantum money schemes and \qefid pairs consists of an oracle sampled from the distribution $\oracle$ that we defined in \Cref{def:S-oracle}. Recall, the oracle $\oracle$ has been useful to establish our previous separation. Specifically, it provided us with the source of hardness needed to build a \qefid pair that is secure against any $\oracle$-aided adversary that makes a bounded number of queries to its oracle, but is otherwise unbounded; see \Cref{lem:computational-indistinguishability-EFID-pair}. Hence, we already have \Cref{item:efid-exist} in hand (namely the existence of statistically-secure \qefid pairs) and it remains to show how to break any private-key quantum money scheme relative to $\oracle$.

To that end, we will use \emph{shadow tomography}~\cite{Shadow2}, which is a powerful tool that, roughly speaking, enables estimating the outcome of exponentially many POVMs on a given quantum state given only polynomially many copies of that states. In particular, shadow tomography can be seen as a strengthening of the gentle search procedure which we used to break any OWSG in \Cref{thm:break-OWSG}.

We will use the following result from~\cite{Shadow2}.

\begin{theorem}[{\cite[Theorem 2, Problem 1]{Shadow2}}]\label{thm:shadow-tomography-estimates}
    Let $\epsilon,\delta\in \mathbb{R}$ and $D,M\in\NN$.
    There exists $k\in\tilde{O}(\log D\cdot \log^4 M\cdot \log(1/\delta) / \epsilon^4)$, where $\tilde{O}$ notation hides a $\poly(\log\log(M),\log\log(D),\log(1/\epsilon))$ factor, and a quantum algorithm that takes as input $k$ copies of a (possibly mixed) $D$-dimensional quantum state $\rho$, as well as a list of binary-values POVMs $\Pi_1,\ldots, \Pi_M$,
    and outputs real numbers $b_1,\ldots,b_M\in [0,1]$ such that, with probability at least $1 - \delta$, for every $i\in [M]$,
    \ifnum\submission=0
    $$|b_i-\Tr(\Pi_i\rho)|\leq \epsilon.$$
    \else
        $|b_i-\Tr(\Pi_i\rho)|\leq \epsilon$.
    \fi
\end{theorem}

{\renewcommand{\QQ}{\mathbf{M}}

Equipped with \Cref{thm:shadow-tomography-estimates}, we are able to demonstrate a statistical attack that breaks any private-key quantum money scheme relative to any oracle $O\in \oracle$.

\begin{theorem}\label{thm:statistical-break-quantum-money}
For any $O\in\cO$ and any private-key quantum money scheme $\QQ$ relative to $O$ that satisfies correctness (as defined in \Cref{def:quantum-money}), there exists a polynomial-query oracle-aided adversary $\adv$ that breaks the unforgeability of $\QQ$.
\end{theorem}

Interestingly, we stress that our statistical attack does not make any queries to the verification oracle. \Cref{thm:statistical-break-quantum-money} in combination with \Cref{thrm:efi-exist}, gives us the following corollary: in any fully black-box construction of private quantum money schemes from \qefid pairs, the reduction either has complexity greater than $2^{\sec/51}$ or success probability smaller than $2^{-\sec/51}$ (see \Cref{sec:fbb-def} for formal definitions). \ifnum\submission=1 We give a proof to this implication in \Cref{sec:qmoney-proofs}.\fi

\begin{restatable}[Formal Restatement of \Cref{thm:efid-vs-money}]{corollary}{efidqmoney}
    \label{cor:efid-qmoney}
    Private-key quantum money schemes are $2^{\sec/51}$-separated from \qefid pairs.
\end{restatable}

\submitorfull{}{efid-vs-qmoney}

The remaining of this section is dedicated to the proof of \Cref{thm:statistical-break-quantum-money}.

\subsection{Proof of \texorpdfstring{\Cref{thm:statistical-break-quantum-money}}{Theorem 6.2}}

Fix $O\in \cO$. Let $\QQ=(\kgen,\mint,\verify)$ be a private-key quantum money scheme relative to $O$. 
Let $\kappa$ is the security parameter for $\QQ$. Let $n:=n(\kappa)$ denote the size (in qubits) of a money state output by $\mint$ on input key of length $\kappa$. Let $\mu:=\mu(\kappa)\in[0,1]$ be an inverse polynomial function representing the correctness of $\QQ$. Since $\mint$ and $\verify$ are QPT algorithms, they only make polynomially many queries to $O$. Let $q:=q(\kappa)$ be a polynomial bound on the query complexity of both $\mint$ and $\verify$. Let $\overline{\sec}:=\overline{\sec}(\kappa)$ be the polynomial bound on the length of the queries made by $\mint$ and $\verify$. That is, for any $\kappa\in\NN$ and any 
$k\in\bin^\kappa$, 
 $\mint(1^\kappa)$ and $\verify(k,\cdot)$ make at most $q(\kappa)$ queries to $O$, where each query is of length at most $\overline{\sec}(\kappa)$. 
Let 
\begin{equation}\label{eq:tau-def}
    \tau:=\bigotimes_{\sec=1}^{\overline{\sec}}\ketbra{S_\sec -}{S_\sec -}^{\otimes 2q^2\kappa^2/\mu^2}.
\end{equation}

We invoke reflection emulation again to ``de-oraclize'' the algorithms $\mint$ and $\verify$: By \Cref{prop:reflect-emulate-efid}, since $O$ consists of reflection oracles and all queries made to $O$ are of length at most $\overline{\secp}$, there exist QPT algorithms $\widetilde{\mint}$ and $\widetilde{\verify}$ such that for every $\kappa\in \NN$, key $k\in \bin^\kappa$ and state $\rho$,\footnote{Note we define $\widetilde{\mint}$ to simulate $\mint$ while outputting the state $\tau$ that it receives as input. By \Cref{prop:reflect-emulate-efid}, the state $\tau$ does not get disturbed by too much. For the simulation of $\verify$, however, we do not care about preserving the state $\tau$, and therefore, we omit it from the output.}
\begin{align}\label{eq:mint-to-tilde-qmoney}
    &\TD[\mint^{O}(k)\otimes\tau,\widetilde{\mint}(k,\tau)]  \leq \frac{\mu}{\kappa},\\
    \label{eq:verify-to-tilde-qmoney}
    &\text{and}\quad \left| \Pr[\verify^{O}(k,\rho)=1] -\Pr[\widetilde{\verify}(k,\rho,\tau)=1] \right| \leq \frac{\mu}{\kappa}.
\end{align}
Let $\VM$ denote the algorithm that takes as input two keys $k,k'\in\bin^\kappa$, for some $\kappa\in\NN$, and a $2\overline{\sec}q^2\kappa^2/\mu^2$-register state $\tau$ (which will be of the form of \Cref{eq:tau-def}), and behaves as follows:
\begin{description}
    \item[] $\VM(k,k',\tau)$:
    \if\submission=0
    \begin{enumerate}
        \item Compute $(\rho,\tau')\gets\widetilde{\mint}(k',\tau)$.
        \item Output $\widetilde{\verify}(k,\rho,\tau')$.\footnote{Note that the states $\rho$ and $\tau$ can be entangled.}
    \end{enumerate}
    \else
    \begin{enumerate*}
        \item Compute $(\rho,\tau')\gets\widetilde{\mint}(k',\tau)$.
        \item Output $\widetilde{\verify}(k,\rho,\tau')$.\footnote{Note that the states $\rho$ and {$\tau'$} can be entangled.}
    \end{enumerate*}
    \fi
\end{description}

Our attack against $\QQ$ consists of two invocation of shadow tomography, as formalized in \Cref{thm:shadow-tomography-estimates}. Let $m:=m(\kappa)=\kappa^{10}\cdot n(\kappa)/{\mu(\kappa)^4}$, which is polynomial since $\mu$ is inverse polynomial in $\kappa$. 

First, we consider the collection of POVM elements $\{\Pi_k\}$ where, for any $k\in\bin^\kappa$, $\Pi_k$ corresponds to that $\widetilde{\verify(k,\cdot)}$ accepts, i,e., outputs $1$. By \Cref{thm:shadow-tomography-estimates}, there exists an algorithm $\estver$ that, for any $\kappa\in\NN$ and any $n(\kappa)$-qubit state $\rho$, takes as input $\left(\rho\otimes \tau\right)^{\otimes m(\kappa)}$
 and outputs estimates $\{b_\V(k,\rho)\}_{k\in \bin^\kappa}$ such that with probability $1-\frac{1}{2^\kappa}$, for every $k\in \bin^\kappa$,
\begin{equation}\label{eq:estimate-shadow-tomography-private-quantum-money}
|b_\V(k,\rho)-\Pr[{\widetilde{\verify}}(k,\rho, \tau))=1]|\leq \frac{\mu(\kappa)}{\kappa}.
\end{equation}

Second, we consider the collection of POVM elements $\{\Lambda_{k,k'}\}$ where, 
for any $(k,k')\in\{\bin^\kappa\times\bin^\kappa\}_{\kappa}$, $\Lambda_{k,k'}$ corresponds to that $\VM(k,k',\cdot)$ outputs $1$. By \Cref{thm:shadow-tomography-estimates} again, there exists an algorithm $\estvm$ that, on input $\tau^{\otimes m}$ outputs estimates $\{b_{\VM}(k,k')\}_{k,k'\in \bin^\kappa}$, such that with probability $1-\frac{1}{2^\kappa}$, for every $k,k'\in \bin^\kappa$,
\begin{equation}\label{eq:estimate-shadow-tomography-vm-private-quantum-money}
|b_{\VM}(k,k')-\Pr[\VM(k,k',\tau))=1]|\leq \frac{\mu(\kappa)}{\kappa}.
\end{equation}

Using $\estver$ and $\estvm$, we propose the following query-efficient forger $\adv$ against $\QQ$. 
For any $\kappa\in\NN$, on input $m(\kappa)$ copies of an $n(\sec)$-qubit money state $\$^*$, 
$\adv$ behaves as follows:
\begin{enumerate}
    \item For $\sec=1,\dots,\overline{\sec}$, $\adv$ runs $E^O(1^\sec)$ (from \cref{cor:phase-estimation}) $6q^2\kappa^2m(\kappa)/\mu(\kappa)^2$ times and obtains $3m$ copies of $\ketbra{S_\sec-}{S_\secp-}^{\otimes 2q^2 \kappa^2/\mu^2}$.
    Overall, this results in the state $\tau^{\otimes 3m}$, where $\tau$ is as defined in \Cref{eq:tau-def}.
    \item $\adv$ runs $\estver$ on $(\$^*\otimes \tau)^{\otimes m}$ and gets back estimates $\{b_\V(k,\$^*)\}_{k\in \bin^\kappa}$.\label{it:estimator-ver}
    \item $\adv$ runs $\estvm$ on $\tau^{\otimes 2m}$ and gets back estimates $\{b_{\VM}(k,k')\}_{k,k'\in \bin^\kappa}$.\label{it:estimator-vm}
    \item \label{step:find} Find the lexicographically first $k'$ such that 
    \begin{equation}\label{eq:constraint-finder}
    |b_{\V}(k,\$^*)-b_{\VM}(k,k')|\leq 5\mu(\kappa)/\kappa
    \end{equation}
    for all $k\in\bin^\kappa$. If such a $k'$ does not exist, abort. 
    \item Let $m'=\ceil{2{\kappa}m/\left(\mu\left(1-{10}/{\kappa}\right)\right)}$\footnote{Here $\ceil{\cdot}$ denotes the ceiling function on real numbers.}. Run $\mint^O(k')$ for $m'$ times and output all resulting states.\label{it:final-step}
    
\end{enumerate}

\Cref{thm:statistical-break-quantum-money} follows by the following lemma, which we fully prove in \Cref{sec:qmoney-attack-proof}.

\begin{restatable}{lemma}{qmoneyattack}\label{lem:qmoney-attack}
    The algorithm $\adv$ described above breaks the unforgeability of the arbitrary private-key quantum money scheme $\QQ$.
\end{restatable}

\submitorfull{}{qmoney-attack}

\ifnum\anonymous=1
\else
\paragraph{Acknowledgments.}
We want to thank Or Sattath for the helpful discussion on the specification of the $\qpspace$ machine oracle.

    \BeforeBeginEnvironment{wrapfigure}{\setlength{\intextsep}{0pt}}
    \begin{wrapfigure}{r}{100px}
        \includegraphics[width=100px]{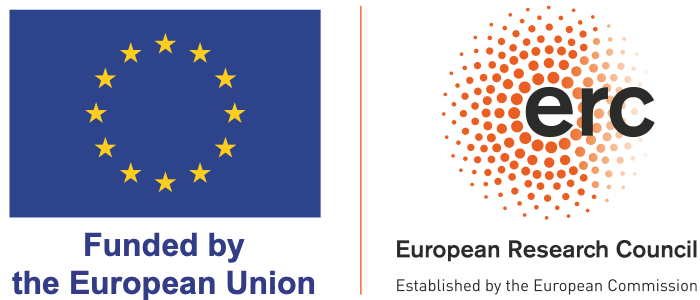}
    \end{wrapfigure}
            
    Amit Behera was funded by the Israel Science Foundation (grant No. 2527/24) and the European Union (ERC-2022-COG, ACQUA, 101087742). Views and opinions expressed are however those of the author(s) only and do not necessarily reflect those of the European Union or the European Research Council Executive Agency. Neither the European Union nor the granting authority can be held responsible for them.

    Giulio Malavolta is supported by the European Research Council through an ERC Starting Grant (Grant agreement No.~101077455, ObfusQation).

    Tomoyuki Morimae is supported by
    JST CREST JPMJCR23I3, 
    JST Moonshot R\verb|&|D JPMJMS2061-5-1-1, 
    JST FOREST, 
    MEXT QLEAP, 
    the Grant-in Aid for Transformative Research Areas (A) 21H05183,
    and 
    the Grant-in-Aid for Scientific Research (A) No.22H00522.

    Tamer Mour is supported by European Research Council (ERC) under the EU’s Horizon 2020 research and innovation programme (Grant agreement No. 101019547)
            
\fi

\ifnum\submission=0
\bibliographystyle{alpha} 
\else
\bibliographystyle{splncs04}
\fi
\bibliography{abbrev3,crypto,reference}

\ifnum\submission=1
\newpage
\vspace{2em}
\begin{center}
\textbf{\Large Supplementary Materials}
\end{center}
\fi

\appendix

\ifnum\submission=1
\section{Complementary Definitions}

\subsection{Defining Fully Black-Box Separation of Private-key Quantum Money Schemes from \qefid pairs}
\else

\section{Defining Fully Black-Box Separation of Private-key Quantum Money Schemes from \qefid pairs}

\fi

\label{sec:fbb-def}

We hereby define a fully black-box construction of private-key quantum money schemes from \qefid pairs. 
Roughly speaking, such a construction is called fully black-box if it makes only a black-box used of the underlying \qefid pair
and, additionally, the security reduction breaks the underlying \qefid pair
while making only black-box access to the presumed adversary against the quantum money scheme.

\begin{definition}[Fully Black-box Construction of Private-key Quantum Money Schemes from \qefid Pairs]\label{def:fbb-money}
Let $q:\NN\to\NN$ and $\epsilon:\NN\to[0,1]$. A \emph{$(q,\epsilon)$-fully black-box construction of private-key quantum money schemes
from \qefid pairs}, with key space $\keyspace=\{\keyspace_\kappa\}_{\kappa\in\NN}$, is a triple of polynomial-query oracle-aided algorithms $(\kgen,\mint,\verify)$ that satisfy the syntax of a private-key quantum money scheme (see \Cref{def:quantum-money}) and an oracle-aided reduction $\reduction$ satisfying the following properties:
\begin{itemize}
    \item \textbf{Construction Correctness:} For any unitary oracle $O$, $(\kgen,\mint,\verify)$ satisfy correctness relative to $O$. That is, there exists an inverse-polynomial function $\mu:\NN\to[0,1]$ such that, for any $\kappa\in\NN$,
    $$\Pr[\verify^O(k,\$_k);k\gets\kgen^O(1^\kappa),\$_k\gets\mint^O(k)]\geq\mu(\kappa).$$
    \item \textbf{Black-box Security Reduction:} For any unitary oracle $O$, any $O$-aided \qefid pair $(D_0^O,D_1^O)$ that is efficiently sampleable and statistically far as required by \Cref{def:efi}, and any non-uniform oracle-aided quantum adversary $\adv$, if there exist polynomials $m,m':\NN\to\NN$ where $m'>m$ such that, for infinitely many $\kappa\in\NN$,
    \begin{align*} 
        \Pr&[\exists S\subseteq [m'(\kappa)], |S|>m, \verify^O(k,\$_i)=1\forall i\in S;\\
        &k\gets \kgen^O(1^\kappa), {\$}_k\gets \mint^O(k), \$_{1,\cdots,m'(\kappa)}\gets\adv^{\verify^O(k,\cdot)}(1^\kappa,\$_k^{\otimes m(\kappa)})]>\frac{1}{2},
    \end{align*}
    (where $\$_{1,\cdots,m'(\kappa)}$ is a state on $m'(\kappa)$ registers and $\$_i$ denotes its $i^{th}$ register), then $\reduction$ breaks the computational indistinguishability of $(D^O_0,D^O_1)$ with advantage $\epsilon$. That is, for infinitely many $\sec\in\NN$,
    \[
        \left| \Pr_{x\gets D^O_0}[\reduction^{\adv,O}(x)=1] - \Pr_{x\gets D^O_1}[\reduction^{\adv,O}(x)=1] \right|\geq \epsilon(\sec).
    \]
    \item \textbf{Reduction Efficiency:} For any $\sec\in\NN$ and $y\in\bin^\sec$, $\reduction^{f,\adv}(y)$ makes at most $q(\sec)$ queries to the oracles $O$ and $\adv$.
\end{itemize}

\end{definition}

We define a fully black-box $\alpha$-separation to embody the impossibility of any fully black-box construction abiding a trade-off (parameterized by $\alpha$) between the complexity of the underlying reduction and its success probability. A larger value of $\alpha$ gives stronger separation and, in particular, superpolynomial $\alpha$ indicates the impossibility of a reduction that is both polynomial time and has non-negligible advantage, as typically required in the traditional cryptographic setting.

\begin{definition}[Black-box Separation of Private-key Quantum Money Schemes From \qefid Pairs]\label{def:fbb-sep}
    We say that private-key quantum money schemes are \emph{$\alpha$-separated} from \qefid pairs, for $\alpha(\sec)>1$, if for any $(q,\epsilon)$-fully black-box construction of private-key quantum money schemes from \qefid pairs, it holds that either
    \begin{enumerate}
        \item $q(\sec)>O(\alpha(\sec))$, or
        \item $\epsilon(\sec)\leq O(1/\alpha(\sec))$.
    \end{enumerate}
\end{definition}

\ifnum\submission=1

\subsection{Additional Related Quantum Cryptography Primitives} \label{sec:more-defs}

\fi

\section{From Absolute-Gap Distinguisher to Positive-Gap Distinguisher}\label{app:abs-to-pos}

\abstopos*

The proof follows by the same reduction as in~\cite[Corollary 32]{STOC:KQST23} based on Yao's distinguishing/predictor lemma~\cite{Yao82} (also see~\cite{BG11}). We repeat the proof below.

\begin{proof}

We define the distinguisher $\badv^O$ to behave as follows upon receiving an input $x$:
\begin{enumerate}
    \item Sample $c\gets\{0,1\}$.
    \item Sample $x'\gets D^O_c$ and run $\adv^O(x')$ to get an outcome $d$. 
    \item Run $\adv^O(x)$ to get an outcome $e$.
    \item Output $c\oplus d\oplus e$.
\end{enumerate}

For any $O\in\oracle$, we define 
\begin{align*}
a(O)\coloneqq \Pr_{x\gets D^O_0}[\adv^O(x)=1] && \text{and} && b(O)\coloneqq \Pr_{x\gets D^O_1}[\adv^O(x)=1]. 
\end{align*}

It is easy to see that, for any fixed $O\in\oracle$, $\Pr[c\oplus d=0]=(1+a(O)-b(O))/{2}$ 
and $\Pr[c\oplus d=1]=(1+b(O)-a(O))/{2}$. Hence we get,
\begin{align*}
    &\Pr_{x\gets D^O_0}[\badv^O(x)=1]-\Pr_{x\gets D^O_1}[\badv^O(x)=1]\\
    &=\Pr[c\oplus d=0]\cdot \left(a(O)-b(O)\right) + \Pr[c\oplus d=1]\cdot \left(b(O)-a(O)\right)\\
    &=\frac{1+a(O)-b(O)}{2}\cdot \left(a(O)-b(O)\right) + \frac{1+b(O)-a(O)}{2}\cdot \left(b(O)-a(O)\right)\\
    &=\left(a(O)-b(O)\right)^2\\
    &=|a(O)-b(O)|^2
\end{align*}

For a last step, we apply the Cauchy-Schwarz inequality as follows
\begin{align*}
    \Expct_{O\gets\oracle}\left[\Pr_{x\gets D_0}[\badv(x)=1] - \Pr_{x\gets D_1}[\badv(x)=1]\right]& = \Expct_{O\gets\oracle}[|a(O)-b(O)|^2]\\
    & \geq \Expct_{O\gets\oracle}[|a(O)-b(O)|]^2\\
    & = \delta^2.
\end{align*}
\end{proof}

\ifnum\submission=1
\input{complementary-proofs}

\else
\section{Proof of Generalized Reflection Emulation}\label{sec:jls-proof-full}

\reflectjls*

\begin{proof}
Assume without loss of generality that $\adv^{Q,R_{\psi}}$ is just a unitary followed by the discarding of some registers, and that the input state is pure.\footnote{If the input state is a mixed state, then we can instead consider a purification of the input state, and instead consider that $A$ acts on the entire purified state while acting as identity on the purification registers.} We define $\badv$ as follows. Let $T$ denote the input register where the copies of $\ket{\psi}$ reside and let $\ket{\theta}_{T}\coloneqq \ket{\psi}^{\otimes \ell}$. $\badv$ behaves the same as $\adv$ does, except that every query to the oracle $R_{\psi}$ on some register $D$ is replaced by applying the reflection about the symmetric subspace, i.e., $\reflect_\symd$, on the registers $D$ and $T$. Here, $\CC^N$ represents the Hilbert space in which the state $\ket{\psi}$ resides and $\symd$ is the symmetric subspace over the $\ell+1$ registers (i.e., the $\ell$ registers of $T$ and the input register $D$) with respect to $\CC^N$.

Clearly, the runtime of $\badv$ is polynomial in $\ell$, $q$, and the runtime of $\adv$, because implementing the reflection about the symmetric subspace can be done in time polynomial in the number of the number of registers, and the number of qubits for each register~\cite[Page 17, 18]{C:JiLiuSon18} and~\cite{BBD+97}.

We will show that the intermediate state of $\adv$ and $\badv$, right after the first $R_\psi$-query made by $\adv$ or (resp.) the simulation thereof is made by $\badv$, are statistically close. A bound on the distance between the final outcome of the two algorithms is then implied by inducting the same argument over all oracle queries.

Let $\ket{\phi}$ be the intermediate state just before the first oracle query to $R_\psi$ or the simulated query under the algorithms $\adv$ or, respectively, $\badv$ (up to this point the algorithms are identical and so is their intermediate state). Let $D$ be the query register on which $\adv$ queries $R_{\psi}$ and $W$ denote the rest of the registers.
We can write $\ket{\phi}$ as $ \sum_{s}c_s\ket{s}_W\otimes \ket{\phi^s}_D$ for some pure states $\ket{\phi^s}$ and amplitudes $c_s$.  

Define

\begin{align*}
\ket{\Psi^s_\adv}&\coloneqq \reflect_{\psi}(\ket{\phi^s}_D)\otimes \ket{\theta}_{T},\ \text{and} \\
\ket{\Psi^s_\badv}&\coloneqq \reflect_\symd(\ket{\phi^s}_D\otimes\ket{\theta}_{T}). 
\end{align*}

Clearly, the state of all the registers, including the $T$ register under the algorithms $\adv$ and $\badv$ are, respectively,
\begin{align*}\ket{\tilde{\Psi}_\adv}=\left(\sum_{s}c_s\ket{s}_W\otimes \ket{\Psi^s_\adv}_{D,T}\right) && \text{and}
 && \ket{\tilde{\Psi}_\badv}=\left(\sum_{s}c_s\ket{s}_W\otimes \ket{\Psi^s_\badv}_{D,T}\right).
\end{align*}

Note that,
\begin{align}\label{eq:bound-on-inner-product}
|\braket{\tilde{\Psi}_\adv}{{\tilde{\Psi}_\badv}}|& =\left|\sum_s |c_s|^2 \braket{\Psi^s_\adv}{\Psi^s_\badv}\right|.
\end{align}

In the proof of~\cite[Theorem 4]{C:JiLiuSon18}, the following was shown:
\begin{claim}\label{claim:JLS-proof-copied}
    For any pre-query state $\ket{\phi}_D$, the corresponding output states $\ket{\Psi_\adv}_{D,T}$ and $\ket{\Psi_\badv}_{D,T}$ defined as
\begin{align*}
\ket{\Psi_\adv}\coloneqq R_{\psi}(\ket{\phi})\otimes \ket{\theta}_{T} && \text{and}
 &&
\ket{\Psi_\badv}\coloneqq \reflect_\symd(\ket{\phi}\otimes\ket{\theta}_{T}), 
\end{align*}
satisfy, $\braket{\Psi_\adv}{\Psi_\badv}\in \RR$ and
\begin{align*}
\braket{\Psi_\adv}{\Psi_\badv}\geq 1-\frac{2}{{\ell+1}},
\end{align*}
which is non-negative since $\ell\geq 4$.
\end{claim}

We defer the proof of the above claim to the sequel. Combining \Cref{claim:JLS-proof-copied} with \Cref{eq:bound-on-inner-product} we conclude that 
\begin{align}\label{eq:bound-on-inner-product-2}
|\braket{\tilde{\Psi}_\adv}{{\tilde{\Psi}_\badv}}|& 
=\left|\sum_s |c_s|^2 \langle\Psi^s_\adv|\Psi^s_\badv\rangle\right|= \sum_s |c_s|^2 |\braket{\Psi^s_\adv}{\Psi^s_\badv}|\geq \min_s\braket{\Psi^s_\adv}{\Psi^s_\badv}=\braket{\Psi^{s_{\min}}_\adv}{\Psi^{s_{\min}}_\badv},
\end{align}
where $s_{\min}=\arg\min_s |\braket{\Psi^s_\adv}{\Psi^s_\badv}|$.

Hence,
\begin{align*}
   \TD\left[\ket{\tilde{\Psi}_\adv},\ket{\tilde{\Psi}_\badv}\right]
   &=\sqrt{1-|\braket{\tilde{\Psi}_\adv}{\tilde{\Psi}_\badv}|^2}\\
   &\leq\sqrt{1-|\braket{\Psi^{s_{\min}}_\adv}{\Psi^{s_{\min}}_\badv}|^2}&\text{By \Cref{eq:bound-on-inner-product-2}.}\\
   &\leq \sqrt{1-\left(1-\frac{2}{\ell+1}\right)^2}&\text{By \Cref{claim:JLS-proof-copied}}\\
    &\leq \sqrt{1-(1-2\cdot\frac{2}{\ell+1})}=\sqrt{\frac{4}{\ell+1}}=\frac{2}{\sqrt{\ell+1}}.
\end{align*}
Let $\ket{\tilde{\Psi}^q_\adv}$ and $\ket{\tilde{\Psi}^q_\badv}$ denote the final states of all registers (including the $T$ registers) of algorithms $\adv$ and $\badv$ before discarding any registers.  Then, by inducting
on the number of oracle queries, 
we conclude that if $\ket{\tilde{\Psi}^q_\adv}$ and $\ket{\tilde{\Psi}^q_\badv}$ denote the state of all registers including the $T$ registers after $q$ queries and just before discarding any registers, under algorithms $\adv$ and $\badv$ respectively, then,
\begin{align*}
\TD\left[\ket{\tilde{\Psi}^q_\adv},\ket{\tilde{\Psi}^q_\badv}\right]\leq \frac{2q}{\sqrt{\ell+1}}.
\end{align*}
Since discarding registers can only reduce the trace distance between two pure states, we conclude that
\begin{align}
\TD\left[\adv^{R_{\psi},Q}(\ket{\phi})\otimes\ket{\psi}^{\otimes \ell}, \badv^Q\left(\ket{\phi}\otimes \ket{\psi}^{\otimes \ell}\right)\right]\leq \frac{2q}{\sqrt{\ell+1}}.
\end{align}

We now recall the proof of \Cref{claim:JLS-proof-copied} for completeness.
First, note that $\reflect_\symd$ can be written as 
\begin{align}\nonumber
\reflect_\symd&=I-2\cdot\mathsf{Proj}_{\symd}\\
&=I-2\cdot\frac{1}{(\ell+1)!}\sum_{\pi\in \mathpzc{S}_{\ell+1}}W_\pi,\label{eq:reflect-symd-permutations}
\end{align}
where $\mathpzc{S}_{\ell+1}$ is the symmetric group over $\ell+1$ elements, 

and $W_\pi:=\sum_{x_1,\ldots,x_{\ell+1}\in \{0,1,\ldots,N-1\}}\ket{x_{\pi^{-1}(1)},\ldots,x_{\pi^{-1}(\ell+1)}}\bra{x_1,\ldots,x_{\ell+1}}$. For a proof of the last equality, see~\cite[Proposition 6]{Har13}.
Next, for any computational basis state $\ket{x}\ket{y}$, note that
\begin{align}\nonumber
    &\bra{x}\otimes\bra{\theta}\reflect_\symd\ket{y}\otimes\ket{\theta}\\\nonumber
    &=\bra{x}\otimes\bra{\theta}\left(I-2\cdot\frac{1}{(\ell+1)!}\sum_{\pi\in \mathpzc{S}_{\ell+1}}W_\pi\right)\ket{y}\otimes\ket{\theta}\\\nonumber
    &=\braket{x}{y}-\frac{2}{(\ell+1)!}\sum_{\pi:\pi(1)= 1}\bra{x}\otimes\bra{\theta}W_\pi\ket{y}\otimes\ket{\theta}-\frac{2}{(\ell+1)!}\sum_{\pi:\pi(1)\neq 1}\bra{x}\otimes\bra{\theta}W_\pi\ket{y}\otimes\ket{\theta}\\\nonumber
    &=\braket{x}{y}-\frac{2}{(\ell+1)!}\sum_{\pi:\pi(1)= 1}\bra{x}\otimes\bra{\psi}^{\otimes \ell}W_\pi\ket{y}\otimes\ket{\psi}^{\otimes\ell}-\frac{2}{(\ell+1)!}\sum_{\pi:\pi(1)\neq 1}\bra{x}\otimes\bra{\psi}^{\otimes \ell}W_\pi\ket{y}\otimes\ket{\psi}^{\otimes \ell}\\\nonumber 
    &=\braket{x}{y}-\frac{2}{(\ell+1)!}\sum_{\pi:\pi(1)= 1}\braket{x}{y}-\frac{2}{(\ell+1)!}\sum_{\pi:\pi(1)\neq 1} \braket{\psi}{y}\braket{x}{\psi}\\\nonumber
    &=\braket{x}{y}-\frac{2\ell!}{(\ell+1)!}\braket{x}{y}-\frac{2((\ell+1)!-\ell!)}{(\ell+1)!} \braket{\psi}{y}\braket{x}{\psi}\\
    &=\frac{\ell-1}{\ell+1}\braket{x}{y}-\frac{2\ell}{\ell+1} \braket{\psi}{y}\braket{x}{\psi}.\label{eq:theta-x-y-jls-proof-elaborated}
\end{align}
Therefore,
we conclude that
\begin{align}\nonumber
    (I\otimes \bra{\theta}) \reflect_\symd (I\otimes \ket{\theta})
    &=(\sum_{x}\ketbra{x}{x}\otimes \bra{\theta}) \reflect_\symd (\sum_{y}\ketbra{y}{y}\otimes \ket{\theta})\\\nonumber
    &=\sum_{x,y}\ket{x}\left(\bra{x}\otimes\bra{\theta}\reflect_\symd\ket{y}\otimes\ket{\theta}\right)\bra{y}\\\nonumber
    &=\sum_{x,y}\left(\frac{\ell-1}{\ell+1}\braket{x}{y}-\frac{2\ell}{\ell+1} \braket{\psi}{y}\braket{x}{\psi}\right)\ketbra{x}{y}&\text{By \Cref{eq:theta-x-y-jls-proof-elaborated}}\\\nonumber
    &=\frac{\ell-1}{\ell+1}\sum_x \ketbra{x}{x} - \frac{2\ell}{\ell+1}\sum_{x,y} \braket{\psi}{y}\braket{x}{\psi}\ketbra{x}{y}\\
    &=\frac{\ell-1}{\ell+1}I - \frac{2\ell}{\ell+1}\ketbra{\psi}{\psi}.\label{eq:JLS-elaborate-operator-product}
\end{align}
Hence,
\begin{align}\nonumber
    \braket{\Psi_\adv}{\Psi_\badv}&=\Tr[(\ket{\phi}\otimes \ket{\theta})(\bra{\phi}\otimes \bra{\theta})\left((R_{\psi}\otimes I)\reflect_\symd\right)]\\\nonumber
    &=\Tr[\ketbra{\phi}{\phi}R_{\psi}\left(I\otimes \bra{\theta}\right)\reflect_\symd\left(I\otimes \ket{\theta}\right)]&\text{By Cyclicity of trace,}\\\nonumber
    &=\Tr[\ketbra{\phi}{\phi}\left(I-2\ketbra{\psi}{\psi}\right)\left(\frac{\ell-1}{\ell+1}I - \frac{2\ell}{\ell+1}\ketbra{\psi}{\psi}\right)]&\text{By \Cref{eq:JLS-elaborate-operator-product}}\\\nonumber
    &=\Tr[\ketbra{\phi}{\phi}\left(\frac{\ell-1}{\ell+1}I-\frac{2\ell}{\ell+1}\ketbra{\psi}{\psi}-\frac{2(\ell-1)}{\ell+1}\ketbra{\psi}{\psi}+\frac{4\ell}{\ell+1}\ketbra{\psi}{\psi}\right)]\\\nonumber
    &=\Tr[\ketbra{\phi}{\phi}\left(\frac{\ell-1}{\ell+1}I+\frac{2}{\ell+1}\ketbra{\psi}{\psi}\right)]\\\nonumber
    &=\frac{\ell-1}{\ell+1} + \frac{2}{\ell+1}|\braket{\psi}{\phi}|^2\\
    &\geq \frac{\ell-1}{\ell+1} =1-\frac{2}{\ell+1}.\label{eq:inner-prod-jls-proof-elaborate}
\end{align}
This concludes the proof of \Cref{claim:JLS-proof-copied}. 

\end{proof}

\fi

\end{document}